\documentclass[journal]{IEEEtran}
\usepackage[T1]{fontenc}
\usepackage{cite}
\usepackage{amssymb}
\usepackage{amsfonts}
\usepackage{fancyhdr}  %
\usepackage{extarrows}
\usepackage{algorithm}
\usepackage{multirow,tabularx}
\usepackage{mathtools}
\usepackage{xcolor}
\usepackage[english]{babel}
\usepackage{caption}
\usepackage{bm}
\usepackage{algorithmic}
\usepackage{amsmath}
\usepackage{subfigure}
\usepackage{makecell}
\usepackage{colortbl}
\usepackage{caption}

\newtheorem{theorem}{Theorem}
\newtheorem{corollary}{Corollary}

\newtheorem{remark}{Remark}

\newenvironment{proof}{{\indent \it Proof:}}{\hfill $\blacksquare$\par}

\ifCLASSINFOpdf
  \usepackage[pdftex]{graphicx}
  \graphicspath{{../pdf/}{../jpeg/}}
  \DeclareGraphicsExtensions{.pdf,.jpeg,.png}
\else
  \usepackage[dvips]{graphicx}
  \graphicspath{{../eps/}}
  \DeclareGraphicsExtensions{.eps}
\fi
\begin{document}

\title{Environment Sensing Considering the Occlusion Effect: A Multi-View Approach}

\author{Xin~Tong,~\IEEEmembership{Student Member,~IEEE, }
        Zhaoyang~Zhang,~\IEEEmembership{Senior Member,~IEEE, }
        Yihan~Zhang,~\\
        Zhaohui~Yang,~\IEEEmembership{Member,~IEEE, }
        Chongwen~Huang,~\IEEEmembership{Member,~IEEE, }
        Kai-Kit~Wong,~\IEEEmembership{Fellow,~IEEE, }
        and~M\'{e}rouane~Debbah,~\IEEEmembership{Fellow,~IEEE}
\thanks{X.~Tong, Z.~Zhang (Corresponding  Author), Y.~Zhang, Z.~Yang and C.~Huang (e-mails: \{tongx, ning\_ming, zhangyihan, yang\_zhaohui, chongwenhuang\}@zju.edu.cn) are with the College of Information Science and Electronic Engineering, Zhejiang University, Hangzhou 310027, China, and the Key Laboratory of Information Processing, Communication and Networking of Zhejiang Province (IPCAN), Hangzhou 310027, China, the Key Laboratory of Collaborative Sensing and Autonomous Unmanned Systems of Zhejiang Province, Hangzhou 310015, China, and the International Joint Innovation Center, Zhejiang University, Haining 314400, China. Z.~Yang is also with Zhejiang Lab, Hangzhou 311121, China.}

\thanks{K.~Wong  (e-mails: \{kai-kit.wong\}@ucl.ac.uk) is with the Department of Electronic and Electrical Engineering, University College London, UK.}

\thanks{M. Debbah is with the Technology Innovation Institute, 9639 Masdar City, Abu Dhabi, United Arab Emirates (email: merouane.debbah@tii.ae) and also with CentraleSupelec, University Paris-Saclay, 91192 Gif-sur-Yvette, France.}


}
\maketitle

\begin{abstract}

  In this paper, we consider the problem of sensing the environment within a wireless cellular framework. Specifically, multiple user equipments (UEs) send sounding signals to one or multiple base stations (BSs) and then a centralized processor retrieves the environmental information from all the channel information obtained at the BS(s). Taking into account the occlusion effect that is common in the wireless context, we make full use of the different views of the environment from different users and/or BS(s), and propose an effective sensing algorithm called GAMP-MVSVR (generalized-approximate-message-passing-based multi-view sparse vector reconstruction). In the proposed algorithm, a multi-layer factor graph is constructed to iteratively estimate the scattering coefficients of the cloud points and their occlusion relationship. In each iteration, the occlusion relationship between the cloud points of the sparse environment is recalculated according to a simple occlusion detection rule, and in turn, used to estimate the scattering coefficients of the cloud points. Our proposed algorithm can achieve improved sensing performance with multi-BS collaboration in addition to the multi-views from the UEs. The simulation results verify its convergence and effectiveness.
  
\end{abstract}


\IEEEpeerreviewmaketitle

\section{Introduction}\label{js}

\subsection{Motivation}
\IEEEPARstart{T}{he} emergence of innovative wireless communication technologies, such as ultra-massive multiple input multiple output (MIMO) technology, intelligent reflective surface (IRS), and wireless artificial intelligence (AI) \cite{LiuIntegrated, Saad, Rappaport}, etc., provides more possibilities for the development of future wireless communication systems design.
In the foreseeable future wireless communication application scenarios, autonomous driving \cite{Kong}, intelligent robot localization \cite{Han}, and unmanned aerial vehicle (UAV) control \cite{Menouar} require not only wireless broadband connections, but also accurate environmental information, including the location, shape, and electromagnetic (EM) characteristics, etc. of objects in the environment. 
Therefore, integrated sensing and communication (ISAC), as a research hotspot for the next generation wireless communication system, aims to use the wireless communication equipment and infrastructure to achieve environment sensing. 

Such kind of environment sensing has long been accomplished by traditional radar technology and its later evolution of joint radar and communication \cite{Liu}.
The latter usually aims to achieve the sharing of software and hardware equipment and time-frequency resources between radar equipment and communication equipment. Compared with traditional radar or early joint radar and communication systems, efficient use of communication signals in the existing communication systems for environment sensing, and efficient exploitation of the environment sensing results to enhance communication, are two major design goals of the ISAC systems. For the former goal, one straightforward realization is to fully exploit the environmental information embedded in the received communication data, due to the simple fact that the distribution of environmental scatterers dramatically affects the wireless multipath channels. For the latter, direct channel prediction (reconstruction) can be conducted based on the sensed distribution and characteristics of scatterers within the environment thus to enhance the performance of communication \cite{Jiaoicc,Jiaotwc,xintong}. In this paper, we mainly focus on the first part and study how to perform efficient and accurate environment sensing based on the uplink data received from multiple users, in a way compatible with existing communication systems thus to achieve a smooth integration of communication and sensing.

One major challenge to do this is the large number of unknown variables embedded in the generally complicated environment, including the unknown location and EM characteristics of the environmental scatterers themselves, and the very common occlusion effect within and between scatterers, which blocks the propagation of EM waves and makes portions of some scatterers or even all of them invisible to the receivers. In such a case, it is hard for a single user or base station (BS) to fully sense the entire environment and the different views of the environment from different users and/or BS(s) should be exploited. However, a challenge still remains since in each separate sensing (transmitting and receiving), we generally cannot know which scatterers or which part of them participates in the final generation of the received signal in an environment yet to know, and it differs from user to user and BS to BS. This motivates us to develop a practical and efficient multi-view sensing algorithm to deal with the above problem.

\subsection{Related Works}\label{rw}
So far, the design of ISAC system has raised extensive research effort. In \cite{Tan}, the author comprehensively discussed the novel applications, key performance requirements, challenges, and future research directions of ISAC system designs. In \cite{AndrewEnabling}, the author compared the difference between the radar sensing scheme and the ISAC scheme. The author showcased that a common signal wave should be jointly designed to realize the integration of sensing and communication, and the advantages and disadvantages of various design methodologies are also discussed.

Rapid progresses have been witnessed in the practical ISAC system and algorithm design recently. \cite{Wild} designed an ISAC system based on the wireless cellular networks. The author jointly designed signal waveforms suitable for sensing and communication based on 5G NR (New Radio) and its frame structure. Also under the 5G framework, \cite{Barneto} proposed a multi-beam based sensing and communication scheme. In addition to the beams used for communication, some other ones are used to sense the environment around the BS. Based on the 4G and 5G frameworks, \cite{Barneto2} studied the processing principles, implementation challenges, and performance of orthogonal frequency division multiplexing (OFDM)-based radars, and achieved excellent sensing performance by using high channel bandwidth and configurable sub-carriers. In a Wi-Fi communication context, based on the Doppler frequency shift embedded in the channel state information (CSI), \cite{Piechocki} achieved the sensing of human body posture. In the regime of internet of vehicles, \cite{Feng} designed a fast beam alignment and tracking algorithm based on a hardware test platform to achieve vehicle tracking while performing millimeter wave communication with low latency and high data rate. As the research interest on ISAC keeps on increasing, more and more potential technologies are also introduced into this field, such as channel modeling \cite{Almers}, joint beam optimization \cite{Chen} and machine learning \cite{Papageorgiou}, etc. 

As for the environment sensing system design, some works have made full use of the sparsity of environmental information and achieved better and lower-complexity solutions based on the compressed sensing (CS) \cite{Donoho, Candes} framework. In the field of microwave computational imaging, the utilization of the intrinsic sparsity of scatterers within the environment is key to their effective detection. This type of algorithm usually treats the scatterer as a sparsely distributed point cloud, thereby convert the imaging problem into a CS reconstruction problem, which is solved by orthogonal matching pursuit (OMP) \cite{Cai}, approximate message passing (AMP) \cite{Donoho2}, generalized approximate message passing (GAMP) \cite{Rangan} and other widely used methods. 

Recently, in \cite{yaojj} and \cite{taoy}, the authors innovatively used BS to sense actively, and based on the block sparse Bayesian learning (BSBL) algorithm, proposed an IRS-assisted microwave computing imaging method. Although the above-mentioned methods effectively exploit the sparsity of the environment but ignore the occlusion effect caused by the scatterers. They ideally believe that EM signals can be freely transmitted to any location in the environment, and there still lacks sufficient consideration of the occlusion effect in the ISAC system design.

\subsection{Main Ideas and Contributions}
In this paper, based on the existing wireless communication framework, we design a multi-view environment sensing scheme taking into full consideration the occlusion effect in an outdoor scenario.

Different from the above-mentioned application scenarios, we have achieved ISAC by noninvasively making use of the multi-user uplink communication data.  At the same time, we have designed an effective sensing algorithm called GAMP-MVSVR which refers to the generalized-approximate-message-passing-based multi-view sparse vector reconstruction.  
In the uplink communication scenario, we consider an environment sensing scenario with single or multiple BSs, where users and BSs perform multi-view and collaborative sensing, and a centralized processor retrieving the environmental information from all the channel information obtained at the BSs. As a special case of the multi-BS scheme, a single BS itself can also complete all the sensing tasks, and we compare the performance differences between the two schemes. 

Our design is depicted as follows. 
In one or more time slots, the BS estimates the multi-user uplink channel after detecting user pilot symbols (sounding signals). We propose an occlusion detection model based on the geometric location of the cloud points and convert the environment sensing problem into a CS reconstruction problem with the occlusion effect. 
Then we propose a probabilistic reasoning model based on the factor graph by analyzing the relationship among multipath channels, environmental information, and occlusion effects.
Based on the well-known AMP method, the proposed GAMP-MVSVR algorithm achieves environment sensing by iteratively estimating the scattering coefficients of the cloud points and their occlusion relationship. In each iteration, the occlusion relationship between the cloud points of the sparse environment is recalculated according to the proposed occlusion detection rule, and in turn, used to estimate the scattering coefficients of the cloud points. Our proposed algorithm achieves improved sensing performance with multi-BS collaboration in addition to the multi-views from the users. To the best of our knowledge, there has been little work on the design of such an ISAC system in the literature.

\begin{figure*}[t]
  \centering
  \includegraphics[width=4.5in]{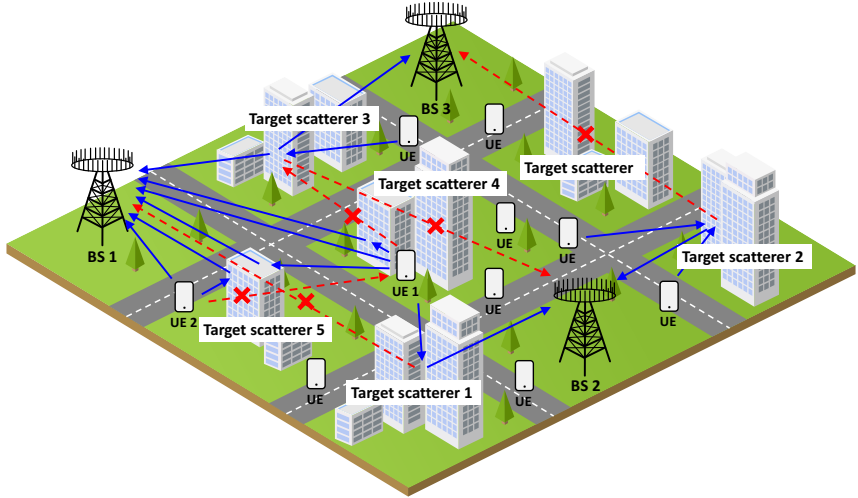}
  \caption{The environment sensing scenario with the occlusion effect.}
  \label{fig1}
  \end{figure*}

The main contributions of this paper are summarized as follows: 

\begin{itemize}
  \item We propose an ISAC scheme, which makes full use of the different views of the environment from different users and/or BS(s) to achieve improved sensing performance with multi-BS collaboration in addition to the multi-views from the UEs.
  
  \item We propose an effective environment sensing algorithm called GAMP-MVSVR that considers the occlusion effect, where a multi-layer factor graph is constructed to iteratively estimate the scattering coefficients of the cloud points and derive their occlusion relationship based on a simple occlusion detection rule.
  
  \item We analyze the impact of the multi-view schemes on environment sensing range and system performance. Extensive simulation results verify the convergence and effectiveness of the proposed algorithm. 
      
\end{itemize}

The rest of this paper is organized as follows. Section \uppercase\expandafter{\romannumeral2} presents the environment setting and system model in the uplink communication scenario. Section \uppercase\expandafter{\romannumeral3} proposes the environment sensing problem formulation. Section \uppercase\expandafter{\romannumeral4} proposes the GAMP-MVSVR algorithm under the occlusion effect. In section \uppercase\expandafter{\romannumeral5}, we analyze the impact of the multi-view schemes on environment sensing range and system performance. Finally, section \uppercase\expandafter{\romannumeral7} presents the numerical results, and section \uppercase\expandafter{\romannumeral8} concludes the paper. 

\textit{Notations}: Fonts $a$, $\bm{a}$ and $\mathbf{A}$ represent scalars, vectors and matrices, respectively. 
Notations $\mathbf{A}^{\rm{T}}$ and $\|\mathbf{A}\|_F$ denote transpose and Frobenius norm of $ \mathbf{A} $, respectively. 
$\mathbf{A}_{i,:}$ and $\mathbf{A}_{:,j}$ represent the $i$-th row and $j$-th column of matrix $\mathbf{A}$, respectively.
$\odot $ represents the Hadamard product between two matrices. 
Finally, notation ${\rm diag}(\bm{a})$ represents a diagonal matrix with the entries of $\bm{a}$ on its main diagonal, and $\delta(\cdot)$ is the Dirac delta function. 

\section{Environment Setting and System Model}
\subsection{Environment Setting}

As shown in Fig. \ref{fig1}, we consider that in the outdoor wireless communication scenario, multiple BSs are deployed and there are multiple active user equipments (UEs). 
There are some buildings (target objects, serve as scatterers) in the outdoor environment that affect the wireless communication channel. 
We assume w.l.o.g. that the channel is a quasi-static channel where one block of transmission time $T_{\rm b}$ is much shorter than the channel coherence time $T_{\rm c}$, $T_{\rm b} << T_{\rm c}$, and the UE is stationary within the coherence time of the multipath channel estimation by UE pilots.
In the uplink communication scenario, multiple single-antenna UEs send uplink communication data to multi-antenna BSs. 
The transmitted signal is scattered by the scatterer, and received by BSs through multipath channel. 
Therefore, the received signal of the BS contains environmental information. 

We consider two schemes: a single-BS scheme and a multi-BS joint scheme to achieve multi-view environment sensing. The specific analysis is as follows.

\begin{itemize}
  \item As shown in Fig. \ref{fig1}, the single-BS environment sensing scheme is treated as a special case of the multi-BS joint scheme.
  When multiple UEs send uplink data to BS 1, the uplink communication channel is affected by the target scatterer.
  In the scenario, due to the occlusion effect, the multipath channel of UE 1 is only affected by the scattering of target scatterers 4 and 5, and the multipath channel of UE 2 is only affected by the scattering of target scatterers 5.
  Therefore, not all scatterers in the environment will affect the multipath channel of the same user, and the occlusion effect makes it hard for a single user to fully sensed the entire environment. The different views of the environment from different users should be exploited to deal with the occlusion effect.

  \item Fig. \ref{fig1} shows the scenario where multiple BSs are used for joint environment sensing, where the multi-BS scheme is performed.
  Sensing by multiple BSs will effectively expand the sensing range. 
  In addition to the same user multi-view method as in the single-BS scheme,  the different views of the environment from different BSs should be exploited.
  For example, in Fig. \ref{fig1}, BS 2 cannot receive the scattered signal from the target scatterer 3, and BS~1 and 3 also have scatterers that cannot be observed. Therefore, a multi-BS joint scheme should be designed to achieve the sensing of the entire environment.
\end{itemize}

\subsection{System Model}\label{mx}
For the receiving antennas deployed on the BS, the single-BS scheme and the multi-BS joint scheme have the same system model.
Let the number of users in the environment be $N_{\rm{u}}$, and the total number of receiving antennas of all BSs is $N_{\rm{R}}$.
In the uplink communication scenario, the channel from the user to the BS is mainly composed of two parts: the line-of-sight channel ${\bf{H}}^{{\rm{LOS}}}$ from the user to the BS, and the multipath channel ${\bf{H}}^{{\rm{S}}}$ caused by the scatterer.
\begin{figure}[t]
  \centering
  \includegraphics[width=2.2in]{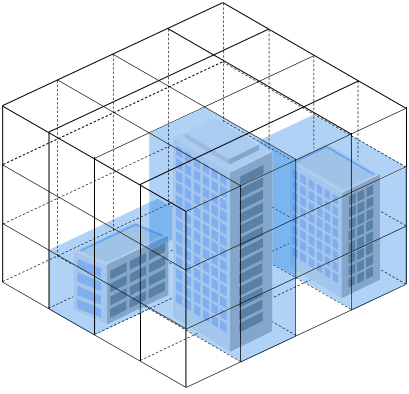}
  \caption{The discretized targeted environment object.}
  \label{fig3}
  \end{figure}
We discretize environmental information and treat the environmental information in the entire environment as point clouds. Each point in the point cloud represents the environmental information of the small cube with length $l_{\rm{s}}$, width $w_{\rm{s}}$, and height $h_{\rm{s}}$ around it, and these small cubes are called pixels.
The inside of each pixel may be empty, or there may be scatterers.
The length, width and height of the environment to be sensed are $L_{\rm{s}}$, $W_{\rm{s}}$ and $H_{\rm{s}}$ respectively, then the number of point clouds is ${N_{\rm{s}}} = \left({{{L_{\rm{s}}}} \mathord{\left/
{\vphantom {{{L_{\rm{s}}}} {{l_{\rm{s}}}}}} \right.
\kern-\nulldelimiterspace} {{l_{\rm{s}}}}} \right)\times \left({{{W_{\rm{s}}}} \mathord{\left/
{\vphantom {{{W_{\rm{s}}}} {{w_{\rm{s}}}}}} \right.
\kern-\nulldelimiterspace} {{w_{\rm{s}}}}} \right)\times \left({{{H_{\rm{s}}}} \mathord{\left/
{\vphantom {{{H_{\rm{s}}}} {{h_{\rm{s}}}}}} \right.
\kern-\nulldelimiterspace} {{h_{\rm{s}}}}}\right)$.
Fig. \ref{fig3} shows the method of discretizing environmental information into point clouds. 
We use the scattering coefficient ${x_{{n_{\rm{s}}}}}$ to represent the scattering coefficient of the small cube where the $n_{\rm{s}}$-th cloud point is located. If the small cube is empty, then ${x_{{n_{\rm{s}}}}} = 0$. Otherwise, ${x_{{n_{\rm{s}}}}} > 0$.
Therefore, the environmental information to be sensed is expressed as 
  ${\bm{x}} = {\left[ {{x_1},{x_2}, \cdots ,{x_{{N_{\rm{s}}}}}} \right]^{\rm{T}}}.$

Multiple users in the space share time-frequency resources. On a single subcarrier, the antennas of all BSs are analyzed independently. The signal received at the $n_{\rm{R}}$-th receiving antenna can be expressed as 
\begin{equation}\small
    {\bf{Y}}\left( {{n_{\rm{R}}}} \right) = {\bf{S}}\left( {{{\bf{H}}^{\rm{S}}}\left( {{n_{\rm{R}}}} \right) + {{\bf{H}}^{{\rm{LOS}}}}\left( {{n_{\rm{R}}}} \right)} \right) + {\bf{W}}, \label{m1a}
\end{equation}
where ${\bf{S}} \in \mathbb{C} ^{N_{\rm{T}} \times N_{\rm{u}}}$ represents codewords with the length of $N_{\rm{T}}$ sent by $N_{\rm{u}}$ users, ${\bf{Y}}\left( {{n_{\rm{R}}}} \right) \in \mathbb{C} ^{N_{\rm{T}} \times 1}$ represents the received signal of the $n_{\rm{R}}$-th BS antenna, and $\bf{W}$ represents noise.
The term ${{\bf{H}}^{{\rm{LOS}}}}\left( {{n_{\rm{R}}}} \right)\in \mathbb{C} ^{N_{\rm{u}} \times 1}$ represents the channel coefficient of the line-of-sight path from the user to the receiving antenna, where the free space channel is decomposed into amplitude attenuation and phase shift, $h^{{\rm{LOS}}}_{n_{\rm{u}}}\left( {{n_{\rm{R}}}} \right) = \alpha^{\rm{ LOS}}\cdot e^{j\varphi^{\rm{LOS}}}$, $\alpha^{\rm{LOS}}$ is the line-of-sight channel amplitude attenuation, and $\varphi^{\rm{LOS }}$ is the line-of-sight channel phase shift. The line-of-sight channel occlusion can be detected through beam scanning, received power, etc.
The term ${{\bf{H}}^{{\rm{S}}}}\left( {{n_{\rm{R}}}} \right)\in \mathbb{C} ^{N_{\rm{u}} \times 1}$ represents the channel coefficient of the multipath channel from the user to the receiving antenna.

As mentioned in Section I-B, we consider that in ISAC system design, the utilization of the intrinsic sparsity of objects or scatterers within an environment is key to the effective detection.
We separate the free space channel (microwave background field) and scatterers in the multipath channel model to obtain sparse environmental information.
In this way, the occlusion relationship between the scatterer and the antenna can be modeled by the existence of the free space channel between them.
The multipath channel from the user to the receiving antenna ${{\bf{H}}^{{\rm{S}}}}\left( {{n_{\rm{R}}}} \right)$ is expressed as
\begin{subequations}\small
  \begin{align} 
      &{{\bf{H}}^{{\rm{S}}}}\left( {{n_{\rm{R}}}} \right) = {\bf{\tilde H}}\left( {{n_{\rm{R}}}} \right){\bm{x}} \label{m1b}\\
      =& \left[ {{\bf{H}}\left( {{n_{\rm{R}}}} \right) \odot {\bf{V}}\left( {{n_{\rm{R}}}} \right)} \right]{\bm{x}} \label{m1c}\\
      =&\left( {{{\bf{H}}^{{\rm{U\rightarrow s}}}}\odot {{\bf{V}}^{{\rm{U\rightarrow s}}}}} \right){\rm{diag}}\left( {{{\bf{H}}^{{\rm{s\rightarrow B}}}}\left( {{n_{\rm{R}}}} \right) \odot {{\bf{V}}^{{\rm{s\rightarrow B}}}}\left( {{n_{\rm{R}}}} \right)} \right){\bm{x}}, \label{m1d}
      \end{align}\label{m1}
\end{subequations}
where the free space channel coefficient ${\bf{\tilde H}}\left( {{n_{\rm{R}}}} \right)$ with occlusion is expressed as the Hadamard product of the free space channel coefficient ${\bf{H}}\left( {{n_{\rm{R}}}} \right)\in \mathbb{C} ^{N_{\rm{u}} \times N_{\rm{s}}}$ and the occlusion matrix ${\bf{V}}\left( {{n_{\rm{R}}}} \right)\in \left\{0,1\right\} ^{N_{\rm{u}} \times N_{\rm{s}}}$.
The element 0 in the occlusion matrix ${\bf{V}}\left( {{n_{\rm{R}}}} \right)$ indicates that the path from $n_{\rm{u}}$-th user to $n_{\rm{R}}$-th receiving antenna is blocked, and the element 1 indicates that the path from $n_{\rm{u}}$-th user to $n_{\rm{R}}$-th receiving antenna is not blocked.
${{\bf{H}}^{{\rm{U\rightarrow s}}}}\in \mathbb{C} ^{ N_{\rm{u}} \times N_{\rm{s}}}$ represents the channel coefficient from the user to the point cloud position, where $h^{{\rm{U\rightarrow s}}}_{n_ {\rm{u}},n_{\rm{s}}} = \alpha^{\rm{U\rightarrow s}}\cdot e^{j\varphi^{\rm{U\rightarrow s}}}$, $\alpha^{\rm{U\rightarrow s}}$ represents the channel amplitude attenuation from the user to the point cloud position, $\varphi^ {\rm{U\rightarrow s}}$ represents the channel phase shift from the user to the point cloud position.
${{\bf{H}}^{{\rm{s\rightarrow B}}}}\left( {{n_{\rm{R}}}} \right)\in \mathbb{C} ^{N_{\rm{s}} \times 1}$ represents the channel coefficient from the point cloud position to the BS receiving antenna, where $h^{{\rm{s\rightarrow B}}}_{n_{\rm{s}}}\left( {{n_{\rm{R}}}} \right) = \alpha^{\rm{s\rightarrow B}}\cdot e^{j\varphi^{\rm{s\rightarrow B}}}$, $\alpha^{\rm{s\rightarrow B}}$ represents the channel amplitude attenuation from the point cloud position to the BS receiving antenna, $\varphi^{\rm{s\rightarrow B}}$ represents the channel phase shift from the point cloud position to the BS receiving antenna, ${\bf{H}}\left( {{n_{\rm{R}}}} \right) = {{\bf{H}}^{{\rm{U\rightarrow s}}}}{\rm{diag}}\left( {{{\bf{H}}^{{\rm{s\rightarrow B}}}}\left( {{n_{\rm{R}}}} \right)} \right)$.
${{\bf{V}}^{{\rm{U\rightarrow s}}}}\in \left\{0,1\right\} ^{N_{\rm{u}} \times N_{\rm{s}}}$ represents the occlusion matrix from the user to the point cloud position, and its distribution reflects the different views from different users. When ${{\bf{V}}^{{\rm{U\rightarrow s}}}}$ contains an all-zero column, it means that all users cannot sense the corresponding pixel, and the corresponding pixel is out of the sensing range.
${{\bf{V}}^{{\rm{s\rightarrow B}}}}\left( {{n_{\rm{R}}}} \right)\in \left\{0,1\right\} ^{N_{\rm{s}} \times 1}$ represents the occlusion matrix from the point cloud position to the receiving antenna, and its distribution reflects the different views from different BSs. When ${{\bf{V}}^{{\rm{s\rightarrow B}}}}\left( {{n_{\rm{R}}}} \right)$ contains a zero element, it means that the receiving antenna cannot sense the corresponding pixel, and the corresponding pixel is also out of the sensing range, ${\bf{V}}\left( {{n_{\rm{R}}}} \right) = {{\bf{V}}^{{\rm{U\rightarrow s}}}}{\rm{diag}}\left( {{{\bf{V}}^{{\rm{s\rightarrow B}}}}\left( {{n_{\rm{R}}}} \right)} \right)$.

In (\ref{m1}), the model is derived as follows. First, in (\ref{m1b}), we decompose the multipath channel ${{\bf{H}}^{\rm{S}}}\left( {{n_{\rm{R}}}} \right)$ in (\ref{m1a}) into free space channel with occlusion ${\bf \tilde H}\left( {{n_{\rm{R}}}} \right)$ and environmental information $\bm x$. Then, in (\ref{m1c}), we decompose free space channel with occlusion ${\bf \tilde H}\left( {{n_{\rm{R}}}} \right)$ into free space channel ${\bf H}\left( {{n_{\rm{R}}}} \right)$ and occlusion matrix ${\bf V}\left( {{n_{\rm{R}}}} \right)$.
Finally, in (\ref{m1d}), we further decompose the free space channel with occlusion ${\bf \tilde H}\left( {{n_{\rm{R}}}} \right)$ into the channels from the user, the base station to the point cloud position ${{\bf{H}}^{{\rm{U\rightarrow s}}}}$, ${{\bf{H}}^{{\rm{s\rightarrow B}}}}\left( {{n_{\rm{R}}}} \right)$ and their corresponding occlusion matrices ${{\bf{V}}^{{\rm{U\rightarrow s}}}}$, ${{\bf{V}}^{{\rm{s\rightarrow B}}}}\left( {{n_{\rm{R}}}} \right)$.
\begin{figure}[t]
  \centering
  \includegraphics[width=2.5in]{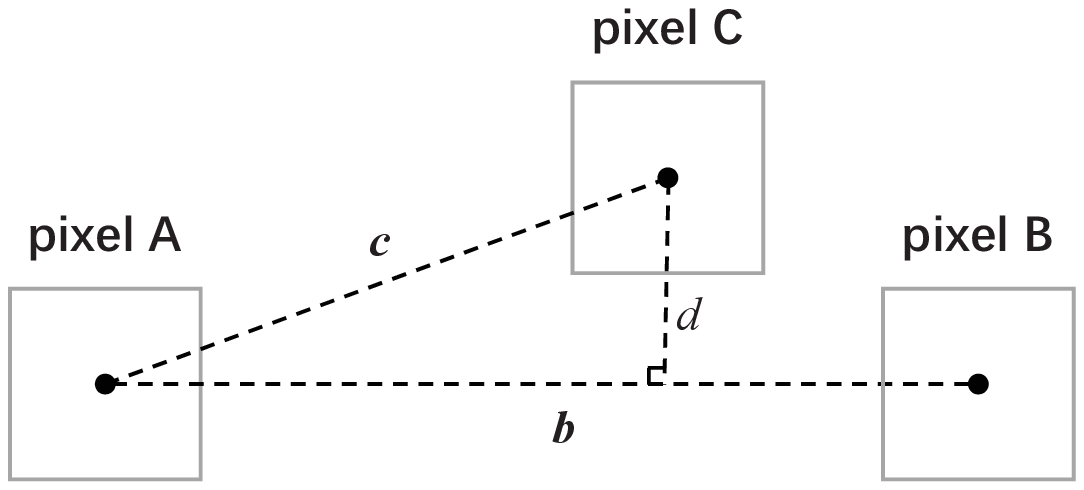}
  \caption{An illustration of the environmental scatterer occlusion detection.}
  \label{fig4}
  \end{figure}

\section{Problem Formulation}
In this section, we analyze the detection method of the occlusion effect in environment sensing, and based on this, we convert the environment sensing problem into a CS reconstruction problem with occlusion effect.

\subsection{Occlusion Detection Model}\label{zd}
We detect the occlusion effect between scatterers based on the geometric positions between pixels.
As shown in Fig. \ref{fig4}, there is a pixel C with a scattering coefficient greater than the threshold $\eta$ between pixels A and B. We believe that pixel C will cause an occlusion effect.
Let the position of pixel A be the origin of coordinates, then the position vectors of pixel B and C are expressed as $\bm{b}$ and $\bm{c}$ respectively.
According to the following conditions, it is considered that pixel C blocks the line-of-sight path between pixels A and B. 

\begin{itemize}
  \item The distance $d$ between pixel C and the line-of-sight path between pixels A and B is less than the threshold $l$,  
  \begin{equation}\small
    d = \frac{{ \left\lvert {{\bm{c}} \times {\bm{b}}}\right\rvert  }}{{ \left\lvert {\bm{b}}\right\rvert  }} < l ,\label{}  
    \end{equation}
where the detection threshold $l$ should be selected according to the side length or diagonal length of the small cube, that is, $l \in \left[\min\left\{l_{\rm{s}},w_{\rm{s}},h_{\rm{s}}\right\},\left(l_{\rm{s}}^2+w_{\rm{s}}^2+h_{\rm{s}}^2\right)^{1/2}\right]$.

  \item The angle between vector $\bm{b}$ and vector $\bm{c}$ is an acute angle, which means that pixel C is located in the same direction of the line-of-sight path from pixel A to pixel B, but not pixel A is located between the pixel B and the pixel C. It is expressed as 
  \begin{equation}\small
    {\bm{b}} \cdot {\bm{c}} > 0.
    \end{equation}

  \item Pixel C is located between pixel A and pixel B, which is expressed as
  \begin{equation}\small
    \left\lvert {{\bm{c}} \cdot {\bm{b}}}\right\rvert   <  \left\lvert {\bm{b}}\right\rvert  ^2.
    \end{equation}

\end{itemize}

According to the above occlusion detection method, we design the environment sensing algorithm. However, the proposed occlusion detection method is not the only method, and the proposed environment sensing algorithm does not depend on the uniqueness of the occlusion detection method.

\subsection{Environment Compressed Sensing Model}\label{ys}
As mentioned in Section \ref{js}, since the distribution of scatterers in the environment is sparse, according to the CS theory, the optimization problem of solving environmental information is expressed as 
\begin{equation}\small
  {\bm{\hat x}} = \arg \mathop {\min }\limits_{\bm{x}} {\left\| {\bm{x}} \right\|_1} \quad {\rm{s}}{\rm{.t}}{\rm{.}}\quad {\left\| {{\bf{Y}} - {\bf{S}}\left( {{{\bf{H}}^{\rm{S}}} + {{\bf{H}}^{{\rm{LOS}}}}} \right)} \right\|_2} \le \varepsilon, \label{q1}  
\end{equation}
where $\varepsilon$ is a slack variable. In the system model of Section II-B, We can calculate the free space channel coefficients ${{\bf{H}}^{{\rm{LOS}}}}\left( {{n_{\rm{ R}}}} \right)$ and ${{\bf{H}}}\left( {{n_{\rm{ R}}}} \right)$ through the basic free space path propagation model.
In addition, some empirical models, simplified models, or interpolation methods are also suitable, and the appropriate model needs to be selected according to the practical environment.
The BS estimates the multipath channel ${{\bf{H}}^{{\rm{S}}}}\left( {{n_{\rm{R}}}} \right)$ by detecting the UE pilot and calculating the relationship with the received signal ${\bf{Y}}\left( {{n_{\rm{R}}}} \right)$. Therefore, in general, in order to achieve accurate channel estimation (e.g. least-squares channel estimation), the pilot length should be greater than the number of users, that is, $N_{\rm T} > N_{\rm u}$.
The proposed algorithm does not rely much on the method of channel estimation. Note that the channel estimation is only the first step of the overall integrated sensing and communication process. The subsequent environment sensing problem studied here is always an underdetermined CS problem with given channel matrix.
Due to the large propagation range of wireless signals, according to different precision requirements, the number of pixels $N_{\rm s}$ will be greater than $N_{\rm u}N_{\rm R}$. At the same time, based on the sparsity of the environmental information $\bm x$, the optimization problem (\ref{q1}) is expressed as the CS reconstruction problem equation
{\small
\begin{align}
  {\left[ {\begin{array}{*{20}{c}}
      {{{\bf{H}}^{\rm{S}}}\left( 1 \right)}\\
       \vdots \\
      {{{\bf{H}}^{\rm{S}}}\left( {{n_{\rm{R}}}} \right)}\\
       \vdots \\
      {{{\bf{H}}^{\rm{S}}}\left( {{N_{\rm{R}}}} \right)}
      \end{array}} \right]_{{N_{\rm{u}}}{N_{\rm{R}}} \times 1}} &= {\left[ {\begin{array}{*{20}{c}}
      {{\bf{\tilde H}}\left( 1 \right)}\\
       \vdots \\
      {{\bf{\tilde H}}\left( {{n_{\rm{R}}}} \right)}\\
       \vdots \\
      {{\bf{\tilde H}}\left( {{N_{\rm{R}}}} \right)}
      \end{array}} \right]_{{N_{\rm{u}}}{N_{\rm{R}}} \times {N_{\rm{s}}}}}{\left[ {\bm{x}} \right]_{{N_{\rm{s}}} \times 1}},\nonumber
  \\ \Rightarrow {{\bf{H}}^{\rm{S}}} &= {\bf{\tilde H}}{\bm{x}}, \label{q2}  
  \end{align}}where ${\bf{\tilde H}}$ is formed by splicing the multipath channel matrix ${{\bf{\tilde H}}\left( {{n_{\rm{R}}}} \right)}$ between $N_{\rm{R}}$ receiving antennas and $N_{\rm{u}}$ users.
In the multi-BS joint environment sensing scheme, $N_{\rm{R}}$ receiving antennas include the number of all antennas on multiple BSs.

\section{Environment Sensing Algorithm under The Occlusion Effect}
In this section, we propose an environment sensing algorithm called GAMP-MVSVR under the occlusion effect to solve the environment information $\bm{x}$ in the (\ref{q2}). According to (\ref{m1}), (\ref{q2}) is expressed as
\begin{equation}\small
  {{\bf{H}}^{\rm{S}}} = {\bf{\tilde H}}{\bm{x}} = \left({\bf{H}}\odot {\bf{V}}\right)  {\bm{x}},\label{ww1} 
\end{equation}
where channel matrix ${\bf{H}}$ is formed by splicing channel matrix ${\bf{H}}\left(n_{\rm{R}}\right)$ of $N_{\rm{R}}$ receiving antennas, and occlusion matrix ${\bf{V}}$ is formed by splicing occlusion matrix ${\bf{V}}\left(n_{\rm{R}}\right)$ of $N_{\rm{R}}$ receiving antennas. 

\begin{figure*}[t]
  \centering
  \includegraphics[width=5.5in]{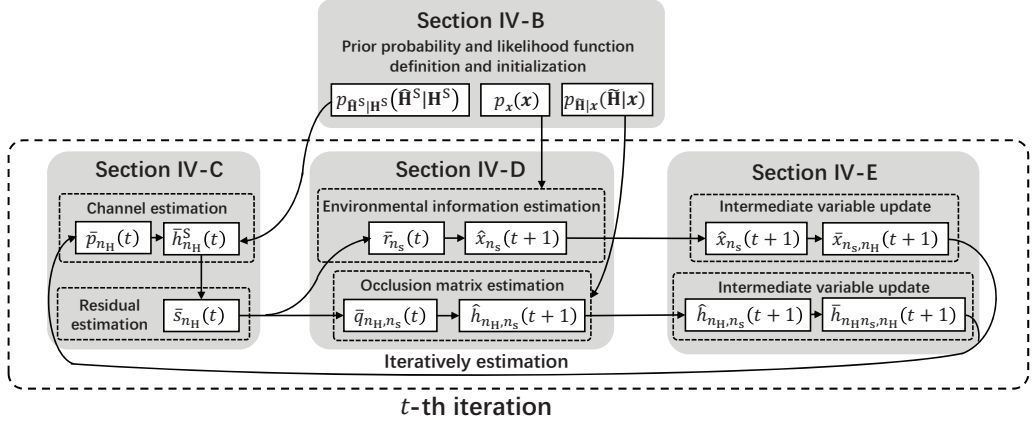}
  \caption{A sketch of the proposed GAMP-MVSVR algorithm.}
  \label{fig6}
  \end{figure*}

\begin{figure}[t]
  \centering
  \includegraphics[width=3in]{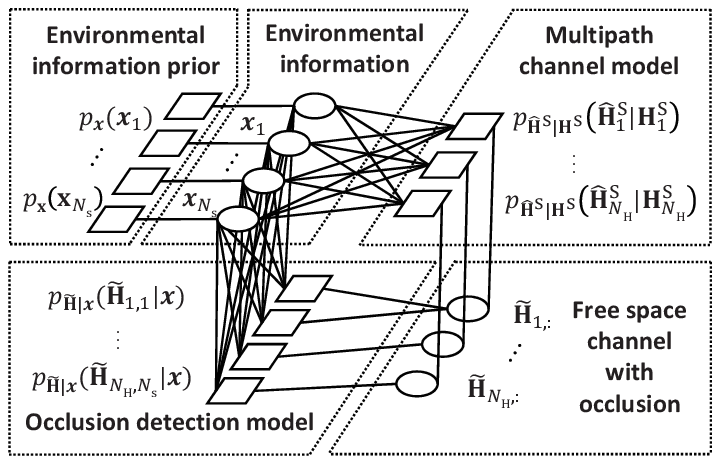}
  \caption{The factor graph of joint posterior distribution.}
  \label{fig5}
  \end{figure}

\subsection{Overview of the Proposed GAMP-MVSVR Algorithm}\label{GAMP2}
Different from the traditional CS reconstruction problem, solving the environment sensing problem with occlusion effect requires not only solving the environment information ${\bm{x}}$ but also needs to consider the unknown occlusion matrix ${\bf{V}}$.
Traditional CS reconstruction algorithms, such as GAMP algorithm \cite{Rangan}, can solve environmental information in an ideal (ignoring occlusion, ${\bf{\tilde H}} = {\bf{H}}$) situation.
In this way, since the accurate measurement matrix ${\bf{\tilde H}}$ cannot be obtained, the environmental information ${\bm{x}}$ cannot be accurately solved.
Compared with the GAMP algorithm, the Bilinear GAMP algorithm \cite{Jason} can solve the bilinear problem, that is, simultaneously solve the occlusion matrix ${\bf{V}}$ and the environmental information ${\bm{x}}$.
However, it does not consider the relationship between the occlusion effect and the environmental information, the Bilinear GAMP algorithm cannot constrain the uniqueness and accuracy of the occlusion matrix ${\bf{V}}$ and the environmental information ${\bm{x}}$.

Fig. \ref{fig6} shows a sketch of the proposed GAMP-MVSVR algorithm.
In Section \ref{SPA}, we propose a factor graph based on the decomposition result of the joint posterior probability.
Based on the sum-product algorithm (SPA) algorithm, we calculate the message passed in the factor graph, and define the prior probability and likelihood function.
In Section \ref{yz2bl}, we approximate the message from factor node to variable node. 
Through the approximate process, we use the result in the previous iteration to calculate its corresponding multipath channel. We compare it with the channel ${{\bf{\hat H}}^{{\rm{S}}}}$ estimated by the pilot and calculate the residual to guide this iteration.
In Section \ref{bl2yz}, we approximate the message from variable node to factor node. 
Through the approximate process, according to the occlusion effect, we estimate ${\bf{\tilde H}}$ and ${\bm{x}}$ based on the residual.
In Section \ref{it}, We update the intermediate variables in Section \ref{yz2bl} and \ref{bl2yz}, and achieve the iterative execution of the estimation methods in Section \ref{yz2bl} and \ref{bl2yz}.

\subsection{Factor Graph and Passed Messages}\label{SPA}
As mentioned in Section \ref{ys}, let ${{\bf{\hat H}}^{\rm{S}}}$ be the multipath channel estimated by the pilot, then ${{\bf{\hat H}}^{\rm{S}}} = {{\bf{H}}^{\rm{S}}} + {\bf{W}_{\rm e}}$, where ${\bf{W}_{\rm e}}$ represents the estimation error.
The estimation error ${\bf{W}_{\rm e}}$ is caused by two main reasons, including the multipath channel estimation error caused by pilot estimation and the free space channel estimation error caused by modeling estimation. In order to simplify the analysis, we assume that the noise ${\bf W}_{\rm e}$ obeys a Gaussian distribution.
The joint posterior distribution of the occluded multipath channel ${\bf{\tilde H}}$ and the environmental information ${\bm{x}}$ is expressed as  
\begin{equation}\small
  \begin{split}
      &{p_{{\bm{x}},{\bf{\tilde H}}|{{\bf{\hat H}}^{\rm{S}}}}}\left( {{\bm{x}},{\bf{\tilde H}}|{{\bf{\hat H}}^{\rm{S}}}} \right)\\
      = & p_{{{\bf{\hat H}}^{\rm{S}}}|{\bm{x}},{\bf{\tilde H}}}\left( {{{\bf{\hat H}}^{\rm{S}}}|{\bm{x}},{\bf{\tilde H}}} \right){p_{\bm{x}}}\left( {\bm{x}} \right){p_{{\bf{\tilde H}}|{\bm{x}}}}\left( {{\bf{\tilde H}}|{\bm{x}}} \right)/{{p_{{\bf{\hat H}}^{\rm{S}}}}\left( {{\bf{\hat H}}^{\rm{S}}} \right)}\\
      \propto& p_{{{\bf{\hat H}}^{\rm{S}}}|{\bf{\tilde H}}{\bm x}}\left( {{{\bf{\hat H}}^{\rm{S}}}|{\bf{\tilde H}}{\bm x}} \right){p_{\bm{x}}}\left( {\bm{x}} \right){p_{{\bf{\tilde H}}|{\bm{x}}}}\left( {{\bf{\tilde H}}|{\bm{x}}} \right)\\
      =& p_{{{\bf{\hat H}}^{\rm{S}}}|{{\bf{H}}^{\rm{S}}}}\left( {{{\bf{\hat H}}^{\rm{S}}}|{{\bf{H}}^{\rm{S}}}} \right){p_{\bm{x}}}\left( {\bm{x}} \right){p_{{\bf{\tilde H}}|{\bm{x}}}}\left( {{\bf{\tilde H}}|{\bm{x}}} \right), \label{a1}  
  \end{split}
\end{equation}
where ${{p_{{\bf{H}}^{\rm{S}}}}\left( {{\bf{H}}^{\rm{S}}} \right)}$ is a constant.
To simplify the expression, let $N_{\rm{H}}={N_{\rm{u}}}{N_{\rm{R}}}$ denote the length of the multipath channel vector ${{\bf{\hat H}}^{\rm{S}}}$.
As shown in Fig. \ref{fig5}, according to the (\ref{a1}), the posterior distribution is expressed in the form of a factor graph.

As shown in Fig. \ref{fig5}, the circles represent the variable nodes $\bf{\tilde H}$ and $\bm{x}$, the squares represent the factor node $p_{{{\bf{\hat H}}^{\rm{S}}}|{{\bf{ H}}^{\rm{S}}}}( {{ {\bf{\hat H}}^{\rm{S}}}|{{\bf{H}}^{\rm{S}}}} )$, ${p_{\bm{x}}}\left( {\bm{x}} \right)$ and ${p_{{\bf{\tilde H}}|{\bm{x}}}}( {{\bf{\tilde H}}|{\bm{x}}} )$.
The lines between the nodes represent that the variables are included in the function of the factor nodes. Specifically, the environmental information prior probability function ${p_{\bm{x}}}\left( {\bm{x}} \right)$ contains the environmental information variable $\bm x$. The relationship between the environmental information variable $\bm x$ and the channel with occlusion $\bf{\tilde H}$ conforms to the occlusion detection model function ${p_{{\bf{\tilde H}}|{\bm{x}}}}( {{\bf{\tilde H}}|{\bm{x}}} )$. The multipath channel observation function $p_{{{\bf{\hat H}}^{\rm{S}}}|{{\bf{ H}}^{\rm{S}}}}( {{ {\bf{\hat H}}^{\rm{S}}}|{{\bf{H}}^{\rm{S}}}} )$ contains the environmental information variable $\bm x$ and the free space channel $\bf{\tilde H}$. The specific function is expressed as follows: 
${{\bf{\hat H}}^{\rm{S}}}$ can be considered as the observation of ${{\bf{H}}^{\rm{S}}}$ under the interference of additive white Gaussian noise, i.e., 
\begin{equation}\small
  p_{{\hat h^{\rm{S}}_{n_{\rm{H}}}}|{h_{n_{\rm{H}}}^{\rm{S}}}}\left( {{\hat h_{n_{\rm{H}}}^{\rm{S}}}|{h_{n_{\rm{H}}}^{\rm{S}}}} \right) = {\mathcal{N}}\left( {\hat h_{n_{\rm{H}}}^{\rm{S}};{h_{n_{\rm{H}}}^{\rm{S}}},{\sigma ^{\rm{w}}}} \right), \label{a2} 
\end{equation}
where ${\sigma ^{\rm{w}}}$ is the variance of Gaussian white noise. Due to the sparseness of environmental information $\bm{x}$, we assume that the $\bm{x}$ obeys a Bernoulli-Gaussian distribution, 
\begin{equation}\small
  {p_{\bm{x}}}\left( {\bm{x}} \right) = \left( {1 - \lambda } \right)\delta \left( x \right) + \lambda {\cal N}\left( {x;{\theta ^{\rm{x}}},{\sigma ^{\rm{x}}}} \right), \label{a3} 
\end{equation}
where $\lambda$ is the sparsity of the environment, $\theta ^{\rm{x}}$ and $\sigma ^{\rm{x}}$ represent the mean and variance of the scattering coefficient of the environmental scatterer. 
${p_{{\bf{\tilde H}}|{\bm{x}}}}\left( {{\bf{\tilde H}}|{\bm{x}}} \right)$ represents the occlusion effect of environmental information $\bm{x}$ on the channel, i.e.,  
\begin{equation}\small
  \begin{split}
  &{p_{{{\tilde h}_{{n_{\rm{H}}},{n_{\rm{s}}}}}|{\bm{x}}}}\left( {{{\tilde h}_{{n_{\rm{H}}},{n_{\rm{s}}}}}|{\bm{x}}} \right) = \\ &\left\{ \begin{array}{l}
      \delta \left( {{{\tilde h}_{{n_{\rm{H}}},{n_{\rm{s}}}}} - {h_{{n_{\rm{H}}},{n_{\rm{s}}}}}} \right)\quad {v_{{n_{\rm{H}}},{n_{\rm{s}}}}} = f_{\rm{occ}}\left( {\bm{x}} \right) = 1\\
      \delta \left( {{{\tilde h}_{{n_{\rm{H}}},{n_{\rm{s}}}}}} \right)\quad \quad \quad \quad \;\;\;{v_{{n_{\rm{H}}},{n_{\rm{s}}}}} = f_{\rm{occ}}\left( {\bm{x}} \right) = 0
      \end{array} \right., 
  \end{split}\label{a4}  
\end{equation}
where $f_{\rm{occ}}\left(\cdot\right)$ represents the occlusion detection function, the specific method is as described in the Section \ref{zd}, and the occlusion detection threshold is set to $\eta = {\theta ^{\rm{x}}}/4$.

Based on the sparse prior information and occlusion relationship of environmental information, we aim to obtain the minimum mean square error (MMSE) estimate of environmental information  
\begin{equation}\small
  \hat x_{{n_{\rm{s}}}}^{{\rm{MMSE}}} = \arg \mathop {\min }\limits_{{{\hat x}_{{n_{\rm{s}}}}}} {\mathbb{E}_{\bm{x}}}\left[ {{{\left( {{x_{{n_{\rm{s}}}}} - {{\hat x}_{{n_{\rm{s}}}}}} \right)}^2}} \right]. \label{a5} 
\end{equation}

According to the loopy belief propagation (LBP) theory \cite{frey1998revolution}, based on the factor graph model, we calculate the MMSE estimated value of the environment information through the SPA \cite{Kschischang}.
As explained in Appendix \ref{fld}, 
the method of SPA message passing is as follows. As shown in Fig. \ref{fig5}, the output message of the variable node along the edge (solid line) is obtained by calculating the product of all its input messages. The output message of the factor node is calculated by multiplying and integrating the function of the factor node and its input message.
The marginal distribution of the variable is calculated from the product of all messages input to the node.
Since it is difficult to accurately calculate the passed message , we approximate the message in the following sections.

\subsection{Approximated Factor-to-Variable Messages}\label{yz2bl}
In this section, we approximate the message from the factor node to the variable node according to the central limit theorem (CLT) and Gaussian approximation based on the idea of approximate message passing \cite{Donoho2,Rangan,Jason}.
Based on CLT, when $N_{\rm{s}}\rightarrow \infty $, $N_{\rm{H}}/N_{\rm{s}}$ remains constant, we drop some high-order infinitesimal terms. 
In the $t$ iteration, the sparsity $\lambda$ in the prior distribution of environmental information ${p_{\bm{x}}}\left( {\bm{x}} \right)$ remains constant. As mentioned in Section \ref{zd} and Section \ref{fx}, the number of zero elements in the occlusion matrix $\bf{V}$ in the variable $\bf{\tilde H}$ remains fixed, and the position of the zero elements is determined by ${\hat x_{{n_{\rm{s}}}}}$. 
Therefore, we consider that variable ${\hat x_{{n_{\rm{s}}}}}$ is independent of variable ${\hat h_{{n_{\rm{H}}},{n_{\rm{s}}}}}$ in the $t$-th iteration.
As explained in Appendix \ref{flb}, the passed message in (\ref{a6}) is approximated as 
{\small
  \begin{align}
      &{\mu _{\hat h_{{n_{\rm{H}}}}^{\rm{S}} \to {x_{{n_{\rm{s}}}}}}}\left( {t,{x_{{n_{\rm{s}}}}}} \right) = \log \int_{h_{{n_{\rm{H}}}}^{\rm{S}}} {{p_{\hat h_{{n_{\rm{H}}}}^{\rm{S}}|h_{{n_{\rm{H}}}}^{\rm{S}}}}\left( {\hat h_{{n_{\rm{H}}}}^{\rm{S}}\left| {h_{{n_{\rm{H}}}}^{\rm{S}}} \right.} \right)} \nonumber\\
       &\times {\cal N}\left( {h_{{n_{\rm{H}}}}^{\rm{S}};{{\bar h}_{{n_{\rm{H}}}{n_{\rm{s}}},{n_{\rm{H}}}}}\left( t \right){x_{{n_{\rm{s}}}}} + {{\bar p}_{{n_{\rm{s}}},{n_{\rm{H}}}}}\left( t \right),}\right.\nonumber\\ & \quad \quad \quad \left.{ \sigma _{{n_{\rm{H}}}{n_{\rm{s}}},{n_{\rm{H}}}}^{\rm{h}}\left( t \right)x_{{n_{\rm{s}}}}^2 + \sigma _{{n_{\rm{s}}},{n_{\rm{H}}}}^{\rm{p}}\left( t \right)} \right) + c\nonumber\\
       &= {F_{{n_{\rm{H}}}}}\left( {{{\bar h}_{{n_{\rm{H}}}{n_{\rm{s}}},{n_{\rm{H}}}}}\left( t \right){x_{{n_{\rm{s}}}}} + {{\bar p}_{{n_{\rm{s}}},{n_{\rm{H}}}}}\left( t \right),}\right.\nonumber\\ & \quad \quad \quad \left.{ \sigma _{{n_{\rm{H}}}{n_{\rm{s}}},{n_{\rm{H}}}}^{\rm{h}}\left( t \right)x_{{n_{\rm{s}}}}^2 + \sigma _{{n_{\rm{s}}},{n_{\rm{H}}}}^{\rm{p}}\left( t \right);\hat h_{{n_{\rm{H}}}}^{\rm{S}}} \right) + c,\label{a17} 
      \end{align}}
where 
\begin{equation}\small
  {F_{{n_{\rm{H}}}}}\left( {{\bar p},{\sigma ^{\rm{p}}};{{\hat h}^{\rm{S}}}} \right) = \log \int_{{h^{\rm{S}}}} {{p_{{{\hat h}^{\rm{S}}}|{h^{\rm{S}}}}}\left( {{{\hat h}^{\rm{S}}}|{h^{\rm{S}}}} \right){\cal N}\left( {{{\hat h}^{\rm{S}}};{\bar p},{\sigma ^{\rm{p}}}} \right)}.\label{a18} 
  \end{equation}
  
Adding the $n_{\rm{s}}$-th term to the summation in ${\bar p_{{n_{\rm{s}}},{n_{\rm{H}}}}}\left( t \right)$ in (\ref{a15}), then ${\bar p_{{n_{\rm{H}}}}}\left( t \right)$ and $\sigma _{{n_{\rm{H}}}}^{\rm{p}}\left( t \right)$ are considered as an estimate of ${h_{{n_{\rm{H}}}}^{\rm{S}}}$, i.e., 
\begin{equation}\small
  {\bar p_{{n_{\rm{H}}}}}\left( t \right) = \sum_{k=1}^{N_{\rm{s}}} {{{\bar h}_{{n_{\rm{H}}}k,{n_{\rm{H}}}}}\left( t \right){{\bar x}_{k,{n_{\rm{H}}}}}\left( t \right)},\label{a19} 
  \end{equation} 
  \begin{equation}\small
  \begin{split}
  \sigma _{{n_{\rm{H}}}}^{\rm{p}}&\left( t \right) = \sum_{k=1}^{N_{\rm{s}}} {\left( {\bar h_{{n_{\rm{H}}}k,{n_{\rm{H}}}}^2\left( t \right)\sigma _{k,{n_{\rm{H}}}}^{\rm{x}}\left( t \right) }\right.}\\ & {\left.{+ \sigma _{{n_{\rm{H}}}k,{n_{\rm{H}}}}^{\rm{h}}\left( t \right)\bar x_{k,{n_{\rm{H}}}}^2\left( t \right) + \sigma _{{n_{\rm{H}}}k,{n_{\rm{H}}}}^{\rm{h}}\left( t \right)\sigma _{k,{n_{\rm{H}}}}^{\rm{x}}\left( t \right)} \right)}.\label{a20}  
  \end{split}
  \end{equation}

We define a residual variable whose mean value is, 
\begin{equation}\small
  {\bar s_{n_{\rm{H}}}}\left( t \right) = {F'_{{n_{\rm{H}}},1}}\left( {{{\bar p}_{{n_{\rm{H}}}}}\left( t \right),\sigma _{{n_{\rm{H}}}}^{\rm{p}}\left( t \right);\hat h_{{n_{\rm{H}}}}^{\rm{S}}} \right),\label{a23} 
  \end{equation}
due to the Gaussian property of ${F_{{n_{\rm{H}}}}}$, the variance of the residual is expressed as, 
\begin{equation}\small
  \sigma _{n_{\rm{H}}}^{\rm{s}}\left( t \right) =  - {F''_{{n_{\rm{H}}},1}}\left( {{{\bar p}_{{n_{\rm{H}}}}}\left( t \right),\sigma _{{n_{\rm{H}}}}^{\rm{p}}\left( t \right);\hat h_{{n_{\rm{H}}}}^{\rm{S}}} \right).\label{a24} 
  \end{equation}
where ${F'_{{n_{\rm{H}}},1}}$ and ${F''_{{n_{\rm{H}}},1}}$ represent the first and second derivatives of the first argument of ${F_{{n_{\rm{H}}}}}$, and ${F'_{{n_{\rm{H}}},2}}$ represents the first derivative of the second argument of ${F_{{n_{\rm{H}}}}}$.

As explained in Appendix \ref{flc}, we get the final approximation result of the passed message ${\mu _{\hat h_{{n_{\rm{H}}}}^{\rm{S}} \to {x_{{n_{\rm{s}}}}}}}\left( {t,{x_{{n_{\rm{s}}}}}} \right)$ in (\ref{a6}) 
{\small
\begin{align}
          &{\mu _{\hat h_{{n_{\rm{H}}}}^{\rm{S}} \to {x_{{n_{\rm{s}}}}}}}\left( {t,{x_{{n_{\rm{s}}}}}} \right) \nonumber \\ & \approx  \left[ {{{\bar s}_{{n_{\rm{H}}}}}\left( t \right){{\bar h}_{{n_{\rm{H}}}{n_{\rm{s}}},{n_{\rm{H}}}}}\left( t \right) + \sigma _{{n_{\rm{H}}}}^{\rm{s}}\left( t \right)\hat h_{{n_{\rm{H}}},{n_{\rm{s}}}}^2\left( t \right){{\hat x}_{{n_{\rm{s}}}}}\left( t \right)} \right]{x_{{n_{\rm{s}}}}} \nonumber \\
          & - \frac{1}{2}\left[ {\sigma _{{n_{\rm{H}}}}^{\rm{s}}\left( t \right)\hat h_{{n_{\rm{H}}},{n_{\rm{s}}}}^2\left( t \right) - \sigma _{{n_{\rm{H}}},{n_{\rm{s}}}}^{\rm{h}}\left( t \right) }\right.\nonumber \\ & \times \left.{ \left( {\bar s_{{n_{\rm{H}}}}^2\left( t \right) - \sigma _{{n_{\rm{H}}}}^{\rm{s}}\left( t \right)} \right)} \right]x_{{n_{\rm{s}}}}^2+c, \label{a26} 
  \end{align}}

According to (\ref{a19}), ${\bar p_{{n_{\rm{H}}}}}$ is considered to be an estimate of ${{\bf{ H}}^{\rm{S}}}$, we can express the posterior probability of ${{\bf{ H}}^{\rm{S}}}$ as the product of the prior probability ${\cal N}\left( {h_{{n_{\rm{H}}}}^{\rm{S}};{{\bar p}_{{n_{\rm{H}}}}}\left( t \right),\sigma _{{n_{\rm{H}}}}^{\rm{p}}\left( t \right)} \right)$ and the likelihood ${p_{\hat h_{{n_{\rm{H}}}}^{\rm{S}}|h_{{n_{\rm{H}}}}^{\rm{S}}}}( {\hat h_{{n_{\rm{H}}}}^{\rm{S}}|h_{{n_{\rm{H}}}}^{\rm{S}}} )$ as follows 
\begin{equation}\small
  \begin{split}
  &{p_{h_{{n_{\rm{H}}}}^{\rm{S}}|{p_{{n_{\rm{H}}}}}}}\left( {h_{{n_{\rm{H}}}}^{\rm{S}}|{{\bar p}_{{n_{\rm{H}}}}}\left( t \right);\sigma _{{n_{\rm{H}}}}^{\rm{p}}\left( t \right)} \right) \\ & = \frac{1}{C}{p_{\hat h_{{n_{\rm{H}}}}^{\rm{S}}|h_{{n_{\rm{H}}}}^{\rm{S}}}}\left( {\hat h_{{n_{\rm{H}}}}^{\rm{S}}|h_{{n_{\rm{H}}}}^{\rm{S}}} \right){\cal N}\left( {h_{{n_{\rm{H}}}}^{\rm{S}};{{\bar p}_{{n_{\rm{H}}}}}\left( t \right),\sigma _{{n_{\rm{H}}}}^{\rm{p}}\left( t \right)} \right),\label{a27} 
  \end{split}
  \end{equation}
where $C=\int_{h^{\rm{S}}}{p_{h_{{n_{\rm{H}}}}^{\rm{S}}|{p_{{n_{\rm{H}}}}}}}\left( {h^{\rm{S}}|{{\bar p}}\left( t \right);\sigma^{\rm{p}}\left( t \right)} \right)$ is a normalization constant. Then the mean and variance of (\ref{a27}) are the mean and variance of calculated ${{\bf{ H}}^{\rm{S}}}$ as follows 
\begin{equation}\small
  \bar{h} _{{n_{\rm{H}}}}^{\rm{S}}\left( t \right) = {\mathbb{E}}\left\{ {h_{{n_{\rm{H}}}}^{\rm{S}}|{{\bar p}_{{n_{\rm{H}}}}}\left( t \right);\sigma _{{n_{\rm{H}}}}^{\rm{p}}\left( t \right)} \right\},\label{a28} 
  \end{equation}
  \begin{equation}\small
  \sigma _{{n_{\rm{H}}}}^{{\rm{Hs}}}\left( t \right) = {\mathbb{D}} \left\{ {h_{{n_{\rm{H}}}}^{\rm{S}}|{{\bar p}_{{n_{\rm{H}}}}}\left( t \right);\sigma _{{n_{\rm{H}}}}^{\rm{p}}\left( t \right)} \right\},\label{a29} 
  \end{equation}
where $\mathbb{E}(X) = \int{xf(x){\rm d}x}$ represents the mean of $x$ and $\mathbb{D}(X) = \int{\left[x-\mathbb{E}(X)\right]^2f(x){\rm d}x}$ represents the variance of $x$.

Pluging (\ref{a27}), (\ref{a28}) and (\ref{a29}) into (\ref{a23}) and (\ref{a24}).
Based on the likelihood in ({\ref{a2}}), according to the derivation in \cite{Rangan,Jason}, the mean and variance of the residual are expressed as 
\begin{equation}\small
  {\bar s_{{n_{\rm{H}}}}}\left( t \right) = \frac{{\hat h_{{n_{\rm{H}}}}^{\rm{S}} - {{\bar p}_{{n_{\rm{H}}}}\left( t \right)}}}{{\sigma _{{n_{\rm{H}}}}^{\rm{p}}\left( t \right) + {\sigma ^{\rm{w}}}}},\label{a32}  
  \end{equation} 
  \begin{equation}\small
      \sigma _{{n_{\rm{H}}}}^{\rm s}\left( t \right) = \frac{1}{{\sigma _{{n_{\rm{H}}}}^{\rm{p}}\left( t \right) + {\sigma ^{\rm{w}}}}}.\label{a33} 
  \end{equation}

As mentioned in this section, in the $t$-th iteration, we consider that variable $\bm{x}$ is independent of variable $\bf{\tilde H}$, so the same as the derivation of (\ref{a26}), we approximate the message ${\mu _{\hat h_{{n_{\rm{H}}}}^{\rm{S}} \to {{\tilde h}_{{n_{\rm{H}}},{n_{\rm{s}}}}}}}( {t,{{\tilde h}_{{n_{\rm{H}}},{n_{\rm{s}}}}}} )$ in (\ref{a8}) as 
{\small
\begin{align}
      &{\mu _{\hat h_{{n_{\rm{H}}}}^{\rm{S}} \to {{\tilde h}_{{n_{\rm{H}}},{n_{\rm{s}}}}}}}\left( {t,{{\tilde h}_{{n_{\rm{H}}},{n_{\rm{s}}}}}} \right) \approx\nonumber\\ &
       \left[ {{{\bar s}_{{n_{\rm{H}}}}}\left( t \right){{\bar x}_{{n_{\rm{s}}},{n_{\rm{H}}}}}\left( t \right) + \sigma _{{n_{\rm{H}}}}^{\rm{s}}\left( t \right)\hat x_{{n_{\rm{s}}}}^2\left( t \right){{\hat h}_{{n_{\rm{H}}},{n_{\rm{s}}}}}\left( t \right)} \right]{{\tilde h}_{{n_{\rm{H}}},{n_{\rm{s}}}}}\nonumber\\ &
      - \frac{1}{2}\left[ {\sigma _{{n_{\rm{H}}}}^{\rm{s}}\left( t \right)\hat x_{{n_{\rm{s}}}}^2\left( t \right) - \sigma _{{n_{\rm{s}}}}^{\rm{x}}\left( t \right) }\right.\nonumber\\ & \times \left.{
       \left( {\bar s_{{n_{\rm{H}}}}^2\left( t \right) - \sigma _{{n_{\rm{H}}}}^{\rm{s}}\left( t \right)} \right)} \right]\tilde h_{{n_{\rm{H}}},{n_{\rm{s}}}}^2 + c.\label{a34}
 \end{align}}

In summary, in this section we approximate the message from factor node to variable node as follows.
We calculate ${{{\bar p}_{{n_{\rm{H}}}}}\left( t \right)}$ (as shown in (\ref{a19})) to obtain the calculated channel ${\bar h _{{n_{\rm{H}}}}^{\rm{S}}\left( t \right)}$ (as shown in (\ref{a28})).
Then we calculate the residual ${\bar s_{{n_{\rm{H}}}}}\left( t \right)$ according to ${\bar h _{{n_{\rm{H}}}}^{\rm{S}}\left( t \right)}$ and the channel estimation value using the pilot ${\hat h_{{n_{\rm{H}}}}^{\rm{S}}}$ (as shown in (\ref{a32})).
Finally, we use the residual and the approximate mean and variance of the passed message to represent the message from the factor node to the variable node (as shown in (\ref{a26}) and (\ref{a34})).

\subsection{Approximated Variable-to-Factor Messages}\label{bl2yz}
In this section, we approximate the message from the factor node to the variable node.
We plug (\ref{a26}) in (\ref{a7}), the message ${{\mu _{{x_r} \to \hat h_{{n_{\rm{H}}}}^{\rm{S}}}}\left( {t+1,{x_r}} \right)}$ from variable node $\bm{x}$ to factor node $p_{{{\bf{\hat H}}^{\rm{S}}}|{{\bf{ H}}^{\rm{S}}}}( {{{\bf{\hat H}}^{\rm{S}}}|{{\bf{H}}^{\rm{S}}}} )$ is approximated as
\begin{equation}\small
  \begin{split}
      &{\mu _{{x_{{n_{\rm{s}}}}} \to \hat h_{{n_{\rm{H}}}}^{\rm{S}}}}\left( {t + 1,{x_{{n_{\rm{s}}}}}} \right)\\
       &= \log {p_{\bm{x}}}\left( {{x_{{n_{\rm{s}}}}}} \right) - \frac{1}{{2\sigma _{{n_{\rm{s}}},{n_{\rm{H}}}}^{\rm{r}}\left( t \right)}}{\left( {{x_{{n_{\rm{s}}}}} - {{\bar r}_{{n_{\rm{s}}},{n_{\rm{H}}}}}\left( t \right)} \right)^2} + c\\
       &= \log \left( {{p_{\bm{x}}}\left( {{x_{{n_{\rm{s}}}}}} \right){\cal N}\left( {{x_{{n_{\rm{s}}}}};{{\bar r}_{{n_{\rm{s}}},{n_{\rm{H}}}}}\left( t \right),\sigma _{{n_{\rm{s}}},{n_{\rm{H}}}}^{\rm{r}}\left( t \right)} \right)} \right) + c,\label{a35}
      \end{split}
  \end{equation}
where
{\small
  \begin{align}
      &{{\bar r}_{{n_{\rm{H}}},{n_{\rm{s}}}}}\left( t \right)\nonumber\\ & = {{\hat x}_{{n_{\rm{s}}}}}\left( t \right)\left( {1 + \sigma _{{n_{\rm{s}}},{n_{\rm{H}}}}^{\rm{r}}\left( t \right)\sum\limits_{k \ne {n_{\rm{H}}}} {\sigma _{k,{n_{\rm{s}}}}^{\rm{h}}\left( t \right)\left[ {\bar s_k^2\left( t \right) - \sigma _k^{\rm{s}}\left( t \right)} \right]} } \right)\nonumber\\
      & + \sigma _{{n_{\rm{s}}},{n_{\rm{H}}}}^{\rm{r}}\left( t \right)\sum\limits_{k \ne n_{\rm{H}}} {{{\hat h}_{k,{n_{\rm{s}}}}}\left( t \right){{\bar s}_k}\left( t \right)},\label{a36}
      \end{align}}
    
\begin{equation}\small
  \begin{split}
  &\sigma _{{n_{\rm{s}}},{n_{\rm{H}}}}^{\rm{r}}\left( t \right)\\ & = {\left( {\sum\limits_{k \ne {n_{\rm{H}}}} {{{\hat h}_{k,{n_{\rm{s}}}}}\left( t \right)\sigma _k^{\rm{s}}\left( t \right) - \sigma _{k,{n_{\rm{s}}}}^{\rm{h}}\left( t \right)\left( {\bar s_k^2\left( t \right) - \sigma _k^{\rm{s}}\left( t \right)} \right)} } \right)^{ - 1}}.\label{a37}
  \end{split}
  \end{equation}
  
Then the mean ${\bar x_{{n_{\rm{s}}},{n_{\rm{H}}}}}\left( {t + 1} \right)$ and variance $\sigma _{{n_{\rm{s}}},{n_{\rm{H}}}}^{\rm{x}}\left( {t + 1} \right)$ of the passed message in (\ref{a35}) are expressed as,
\begin{equation}\small
  \begin{split}
      &{{\bar x}_{{n_{\rm{s}}},{n_{\rm{H}}}}}\left( {t + 1} \right)\\
      & = \frac{1}{C}\int_{{x_{{n_{\rm{s}}}}}} {{x_{{n_{\rm{s}}}}}{p_{\bm{x}}}\left( {{x_{{n_{\rm{s}}}}}} \right){\cal N}\left( {{x_{{n_{\rm{s}}}}};{{\bar r}_{{n_{\rm{H}}},{n_{\rm{s}}}}}\left( t \right),\sigma _{{n_{\rm{H}}},{n_{\rm{s}}}}^{\rm{r}}\left( t \right)} \right)} \\
      & = {G_{\bm{x}}}\left( {{{\bar r}_{{n_{\rm{H}}},{n_{\rm{s}}}}}\left( t \right),\sigma _{{n_{\rm{H}}},{n_{\rm{s}}}}^{\rm{r}}\left( t \right)} \right),\label{a38}
      \end{split}
  \end{equation}
\begin{equation}\small
    \sigma _{{n_{\rm{s}}},{n_{\rm{H}}}}^{\rm{x}}\left( {t + 1} \right) = \sigma _{{n_{\rm{H}}},{n_{\rm{s}}}}^{\rm{r}}\left( t \right){G'_{\bm{x}}}\left( {{{\bar r}_{{n_{\rm{H}}},{n_{\rm{s}}}}}\left( t \right),\sigma _{{n_{\rm{H}}},{n_{\rm{s}}}}^{\rm{r}}\left( t \right)} \right),\label{a39}
  \end{equation}
where $C=\int_{{x_{{n_{\rm{s}}}}}} {{p_{\bm{x}}}\left( {{x_{{n_{\rm{s}}}}}} \right){\cal N}\left( {{x_{{n_{\rm{s}}}}};{{\bar r}_{{n_{\rm{H}}},{n_{\rm{s}}}}}\left( t \right),\sigma _{{n_{\rm{H}}},{n_{\rm{s}}}}^{\rm{r}}\left( t \right)} \right)}$ is a normalization constant, and ${G'_{\bm{x}}}$ represents the first derivative of the first argument of the ${{G}_{\bm{x}}}$.
The specific derivation of the relationship between (\ref{a38}) and (\ref{a39}) was shown in \cite{Rangan}.

Adding the $n_{\rm{H}}$-th term to the summation in ${{\bar r}_{{n_{\rm{H}}},{n_{\rm{s}}}}}\left( t \right)$, then ${{\bar r}_{{n_{\rm{s}}}}}\left( t \right)$ and $\sigma _{{n_{\rm{s}}}}^{\rm{r}}\left( t \right)$ are expressed as
  {\small
  \begin{align}
      {{\bar r}_{{n_{\rm{s}}}}}\left( t \right) & = {{\hat x}_{{n_{\rm{s}}}}}\left( t \right)\left( {1 + \sigma _{{n_{\rm{s}}}}^{\rm{r}}\left( t \right)\sum_{k=1}^{N_{\rm{H}}} {\sigma _{k,{n_{\rm{s}}}}^{\rm{h}}\left( t \right)\left[ {\bar s_k^2\left( t \right) - \sigma _k^{\rm{s}}\left( t \right)} \right]} } \right)\nonumber\\
      & + \sigma _{{n_{\rm{s}}}}^{\rm{r}}\left( t \right)\sum\limits_{k=1}^{N_{\rm{H}}} {{{\hat h}_{k,{n_{\rm{s}}}}}\left( t \right){{\bar s}_k}\left( t \right)},\label{a40}
      \end{align}}
      
{\small
\begin{align}
  \sigma _{{n_{\rm{s}}}}^{\rm{r}}\left( t \right) & = {\left( {\sum_{k=1}^{N_{\rm{H}}} {{{\hat h}_{k,{n_{\rm{s}}}}}\left( t \right)\sigma _k^{\rm{s}}\left( t \right) }}\right.}{\left. {{- \sigma _{k,{n_{\rm{s}}}}^{\rm{h}}\left( t \right)\left( {\bar s_k^2\left( t \right) - \sigma _k^{\rm{s}}\left( t \right)} \right)} } \right)^{ - 1}},\label{a41}
  \end{align}}
  
Based on AMP \cite{Donoho2,Rangan,Jason}, ${{\bar r}_{{n_{\rm{s}}}}}\left( t \right)$ and $\sigma _{{n_{\rm{s}}}}^{\rm{r}}\left( t \right)$ are considered as an estimate of $x_{n_{\rm{s}}}$ as follows
\begin{equation}\small
  {\hat x}_{n_{\rm{s}}}\left( t+1 \right) = {G_{\bm{x}}}\left( {{{\bar r}_{{n_{\rm{s}}}}}\left( t \right),\sigma _{{n_{\rm{s}}}}^{\rm{r}}\left( t \right)} \right),\label{a42}
  \end{equation}
\begin{equation}\small
  \sigma _{{n_{\rm{s}}}}^{\rm{x}}\left( t+1 \right)  = \sigma _{{n_{\rm{s}}}}^{\rm{r}}\left( t \right){G'_{\bm{x}}}\left( {{{\bar r}_{{n_{\rm{s}}}}}\left( t \right),\sigma _{{n_{\rm{s}}}}^{\rm{r}}\left( t \right)} \right),\label{a43}
  \end{equation}
where ${G'_{\bm{x}}}$ represent the first derivative of the first argument of ${G_{\bm{x}}}$. Finally, we calculate the relationship between the mean of the passed message ${\bar x_{{n_{\rm{s}}},{n_{\rm{H}}}}}\left( {t + 1} \right)$ and the mean of the marginal distribution ${\hat x_{{n_{\rm{s}}}}}\left( {t + 1} \right)$ as follows
\begin{equation}\small
  \begin{split}
      &{{\bar x}_{{n_{\rm{s}}},{n_{\rm{H}}}}}\left( {t + 1} \right)\\ &
       \approx {G_{\bm{x}}}\left( {{{\bar r}_{{n_{\rm{s}}}}}\left( t \right) - \sigma _{{n_{\rm{s}}}}^{\rm{r}}\left( t \right){{\hat h}_{{n_{\rm{H}}},{n_{\rm{s}}}}}\left( t \right){{\bar s}_{{n_{\rm{H}}}}}\left( t \right),\sigma _{{n_{\rm{s}}}}^{\rm{r}}\left( t \right)} \right)\\ &
       \approx {G_{\bm{x}}}\left( {{{\bar r}_{{n_{\rm{s}}}}}\left( t \right),\sigma _{{n_{\rm{s}}}}^{\rm{r}}\left( t \right)} \right)\\ &
      - \sigma _{{n_{\rm{s}}}}^{\rm{r}}\left( t \right){{\hat h}_{{n_{\rm{H}}},{n_{\rm{s}}}}}\left( t \right){{\bar s}_{{n_{\rm{H}}}}}\left( t \right){{G'}_{\bm{x}}}\left( {{{\bar r}_{{n_{\rm{s}}}}}\left( t \right),\sigma _{{n_{\rm{s}}}}^{\rm{r}}\left( t \right)} \right)\\ &
       = {{\hat x}_{{n_{\rm{s}}}}}\left( {t + 1} \right) - {{\hat h}_{{n_{\rm{H}}},{n_{\rm{s}}}}}\left( t \right){{\bar s}_{{n_{\rm{H}}}}}\left( t \right)\sigma _n^{\rm{x}}\left( {t + 1} \right),\label{a44}
      \end{split}
  \end{equation}
in which we drop infinitesimal terms such as $\sigma _{{n_{\rm{s}}},{n_{\rm{H}}}}^{\rm{r}}\left( t \right) - {\sigma _{{n_{\rm{s}}}}^{\rm{r}}\left( t \right)}$ in the second step of approximation, and we perform Taylor series expansion on the first argument of the function at point ${{\bar r}_{{n_{\rm{s}}}}}\left( t \right)$ and drop the infinitesimal term in the second step of approximation.

In the $t$-th iteration, we consider that variable $\bm{x}$ is independent of variable $\bf{\tilde H}$ and recalculate the relationship between variables $\bf{\tilde H}$ and $\bm{x}$ after each iteration. Therefore, we approximate the passed message from variable node $\bf{\tilde H}$ to factor node $p_{{{\bf{\hat H}}^{\rm{S}}}|{{\bf{ H}}^{\rm{S}}}}( {{{\bf{\hat H}}^{\rm{S}}}|{{\bf{H}}^{\rm{S}}}} )$ the same way as (\ref{a35}), and calculate the mean ${\hat h_{{n_{\rm{H}}},{n_{\rm{s}}}}}\left( {t + 1} \right)$ and the variance $\sigma _{{n_{\rm{H}}},{n_{\rm{s}}}}^{\rm{h}}\left( {t + 1} \right)$ of the marginal distribution of variable $\bf{\tilde H}$ the same way as (\ref{a42}) and (\ref{a43}), i.e.,
\begin{equation}\small
  {\hat h_{{n_{\rm{H}}},{n_{\rm{s}}}}}\left( {t + 1} \right) = {G_{\bf{H}}}\left( {{{\bar q}_{{n_{\rm{H}}},{n_{\rm{s}}}}}\left( t \right),\sigma _{{n_{\rm{H}}},{n_{\rm{s}}}}^{\rm{q}}\left( t \right)} \right),\label{a45}
  \end{equation}
  \begin{equation}\small
  \sigma _{{n_{\rm{H}}},{n_{\rm{s}}}}^{\rm{h}}\left( {t + 1} \right) = \sigma _{{n_{\rm{H}}},{n_{\rm{s}}}}^{\rm{q}}\left( t \right){G'_{\bf{H}}}\left( {{{\bar q}_{{n_{\rm{H}}},{n_{\rm{s}}}}}\left( t \right),\sigma _{{n_{\rm{H}}},{n_{\rm{s}}}}^{\rm{q}}\left( t \right)} \right),\label{a46}
  \end{equation}
\begin{equation}\small
    \begin{split}
  &{G_{\bf{H}}}\left( {{{\bar q}_{{n_{\rm{H}}},{n_{\rm{s}}}}}\left( t \right),\sigma _{{n_{\rm{H}}},{n_{\rm{s}}}}^{\rm{q}}\left( t \right)} \right) \\ &\quad \quad \quad  \quad \quad = \frac{1}{C}\int_{\tilde h} {{\tilde h}_{{n_{\rm{H}}},{n_{\rm{s}}}}}{p_{{{\tilde h}_{{n_{\rm{H}}},{n_{\rm{s}}}}}|{\bm{x}}}}\left( {{{\tilde h}_{{n_{\rm{H}}},{n_{\rm{s}}}}}|{\bm{x}}} \right)\\ & \quad \quad \quad \quad \quad \times {\cal N} \left( {{{\tilde h}_{{n_{\rm{H}}},{n_{\rm{s}}}}};{{\bar q}_{{n_{\rm{H}}},{n_{\rm{s}}}}}\left( t \right),\sigma _{{n_{\rm{H}}},{n_{\rm{s}}}}^{\rm{q}}\left( t \right)} \right),\label{a47}
    \end{split}
  \end{equation}
where $C=\int_{\tilde h} {{p_{{{\tilde h}}|{\bm{x}}}}( {{{\tilde h}}|{\bm{x}}}) {\cal N}} ( {{{\tilde h}};{{\bar q}}\left( t \right),\sigma^{\rm{q}}\left( t \right)}) $ is a normalization constant, ${G'_{\bf{H}}}$ represent the first derivative of the first argument of ${G_{\bf{H}}}$, ${\bar q_{{n_{\rm{H}}},{n_{\rm{s}}}}}\left( t \right)$ and $\sigma _{{n_{\rm{H}}},{n_{\rm{s}}}}^{\rm{q}}\left( t \right)$ are considered as an estimate of ${{\tilde h}_{{n_{\rm{H}}},{n_{\rm{s}}}}}$ as follows
{\small \begin{align}
  & {\bar q_{{n_{\rm{H}}},{n_{\rm{s}}}}}\left( t \right) = {\hat h_{{n_{\rm{H}}},{n_{\rm{s}}}}}\left( t \right)\left( {1 + \sigma _{{n_{\rm{H}}},{n_{\rm{s}}}}^{\rm{q}}\left( t \right)\sigma _{{n_{\rm{s}}}}^{\rm{x}}\left( t \right)}\right.\nonumber\\ & \left.{
  \left( {\bar s_{{n_{\rm{H}}}}^2\left( t \right) - \sigma _{{n_{\rm{H}}}}^{\rm{s}}\left( t \right)} \right) + \sigma _{{n_{\rm{H}}},{n_{\rm{s}}}}^{\rm{q}}\left( t \right){{\bar x}_{{n_{\rm{s}}},{n_H}}}\left( t \right){{\bar s}_{{n_H}}}\left( t \right)} \right),\label{a48}
  \end{align}}
  \begin{equation}\small
  \sigma _{{n_{\rm{H}}},{n_{\rm{s}}}}^{\rm{q}}\left( t \right) = {\left( {\hat x_{{n_{\rm{s}}}}^2\left( t \right)\sigma _{{n_{\rm{H}}}}^{\rm{s}}\left( t \right) - \sigma _{{n_{\rm{s}}}}^{\rm{x}}\left( t \right)\left( {\bar s_{{n_{\rm{H}}}}^2\left( t \right) - \sigma _{{n_{\rm{H}}}}^{\rm{s}}\left( t \right)} \right)} \right)^{ - 1}}.\label{a49}
  \end{equation}

Similar to the approximation process of (\ref{a44}), we calculate the relationship between the mean of the passed message ${\bar h_{{n_{\rm{H}}}{n_{\rm{s}}},{n_{\rm{H}}}}}\left( {t + 1} \right)$ and the mean of the marginal distribution ${\hat h_{{n_{\rm{H}}},{n_{\rm{s}}}}}\left( {t + 1} \right)$ as follows
\begin{equation}\small
  {\bar h_{{n_{\rm{H}}}{n_{\rm{s}}},{n_{\rm{H}}}}}\left( {t + 1} \right) = {\hat h_{{n_{\rm{H}}},{n_{\rm{s}}}}}\left( {t + 1} \right) - {\hat x_{{n_{\rm{s}}}}}\left( t \right){\bar s_{{n_{\rm{H}}}}}\left( t \right)\sigma _{{n_{\rm{H}}},{n_{\rm{s}}}}^{\rm{h}}\left( t \right).\label{a50}
  \end{equation}

Since the environmental information $\bm{x}$ is an unknown variable, in the $t+1$-th iteration, we use ${\bm{\hat x}}\left( t \right)$ instead of $\bm{x}$ for occlusion detection, (\ref{a47}) is expressed as
{\small \begin{align}
  & {G_{\bf{H}}}\left( {{{\bar q}_{{n_{\rm{H}}},{n_{\rm{s}}}}}\left( t \right),\sigma _{{n_{\rm{H}}},{n_{\rm{s}}}}^{\rm{q}}\left( t \right)} \right) = \frac{1}{C}\int_{\tilde h} {{{\tilde h}_{{n_{\rm{H}}},{n_{\rm{s}}}}}} \nonumber\\ & { {p_{{{\tilde h}_{{n_{\rm{H}}},{n_{\rm{s}}}}}|{{\bm{\hat x}}\left( t \right)}}}\left( {{{\tilde h}_{{n_{\rm{H}}},{n_{\rm{s}}}}}|{{\bm{\hat x}}\left( t \right)}} \right){\cal N}} \left( {{{\tilde h}_{{n_{\rm{H}}},{n_{\rm{s}}}}};{{\bar q}_{{n_{\rm{H}}},{n_{\rm{s}}}}}\left( t \right),\sigma _{{n_{\rm{H}}},{n_{\rm{s}}}}^{\rm{q}}\left( t \right)} \right).\label{a51}
  \end{align}}

In summary, in this section we approximate the message from variable node to factor node as follows.
Firstly, we calculate ${{\bar r}_{{n_{\rm{s}}}}}\left( t \right)$ and ${\bar q_{{n_{\rm{H}}},{n_{\rm{s}}}}}\left( t \right)$ based on the residual ${{{\bar s}_{{n_H}}}\left( t \right)}$ (as shown in (\ref{a40}) and (\ref{a48})), and calculate the new estimated values ${\hat x_{{n_{\rm{s}}}}}\left( {t + 1} \right)$ and ${\hat h_{{n_{\rm{H}}},{n_{\rm{s}}}}}\left( {t + 1} \right)$ (as shown in (\ref{a42}) and (\ref{a45})).
Secondly, since the message from the variable node to the factor node is similar to the posterior probability, we approximate the relationship between the estimated value ${\hat x_{{n_{\rm{s}}}}}\left( {t + 1} \right)$, ${\hat h_{{n_{\rm{H}}},{n_{\rm{s}}}}}\left( {t + 1} \right)$ and the passed message ${\bar x_{{n_{\rm{s}}},{n_{\rm{H}}}}}\left( {t + 1} \right)$, ${\bar h_{{n_{\rm{H}}}{n_{\rm{s}}},{n_{\rm{H}}}}}\left( {t + 1} \right)$ (as shown in (\ref{a44}) and (\ref{a50})).

\begin{algorithm}[t]
  \small
  \caption{Proposed GAMP-MVSVR Algorithm}\label{sf1}
  \begin{algorithmic}[1]
  \REQUIRE
    Given the free space channel coefficient $\bf{H}$ and the estimated value of the uplink multipath channel ${{\bf{\hat H}}^{\rm{S}}}$.
  \STATE
    \textbf{Initialization}: For each ${n_{\rm{H}}}$ and ${n_{\rm{s}}}$, set the prior probability ${p_{\bm{x}}}\left( {\bm{x}} \right)$ of environmental information $\bm{x}$ in (\ref{a3}), and choose ${\hat x_{{n_{\rm{s}}}}\left( 1 \right)}$ and ${\sigma _{{n_{\rm{s}}}}^{\rm{x}}\left( 1 \right)}$.
    Set the likelihood $p_{{\hat h^{\rm{S}}_{n_{\rm{H}}}}|{h_{n_{\rm{H}}}^{\rm{S}}}}({{\hat h_{n_{\rm{H}}}^{\rm{S}}}|{h_{n_{\rm{H}}}^{\rm{S}}}})$ in (\ref{a2}). 
    Set ${\bar s_{{n_{\rm{H}}}}}\left( 0 \right)=0$.
  \STATE
    For each ${n_{\rm{H}}}$ and ${n_{\rm{s}}}$, plug ${\hat x_{{n_{\rm{s}}}}\left( 1 \right)}$ into (\ref{a4}), calculate the free space channel coefficient distribution ${p_{{{\tilde h}_{{n_{\rm{H}}},{n_{\rm{s}}}}}|{\bm{x}}\left(1\right)}}( {{{\tilde h}_{{n_{\rm{H}}},{n_{\rm{s}}}}}|{\bm{x}}\left(1\right)} )$ with occlusion, and choose ${{{\hat h}_{{n_{\rm{H}}},{n_{\rm{s}}}}}\left( t \right)}$ and $\sigma _{{n_{\rm{H}}},{n_{\rm{s}}}}^{\rm{h}}\left( t \right)$.
  \FOR {$t=1,2,\ldots,T_{\rm{max}}$}
  \STATE 
    For each ${n_{\rm{H}}}$, plug ${\hat x_{{n_{\rm{s}}}}\left( t \right)}$, ${\sigma _{{n_{\rm{s}}}}^{\rm{x}}\left( t \right)}$, ${{{\hat h}_{{n_{\rm{H}}},{n_{\rm{s}}}}}\left( t \right)}$, and $\sigma _{{n_{\rm{H}}},{n_{\rm{s}}}}^{\rm{h}}\left( t \right)$ into (\ref{a52}) and (\ref{a53}) to obtain ${\bar p_{{n_{\rm{H}}}}}\left( t \right)$ and $\sigma _{n_{\rm{H}}}^{\rm{p}}\left( t \right)$.
  \STATE 
    For each ${n_{\rm{H}}}$, plug ${\bar p_{{n_{\rm{H}}}}}\left( t \right)$, $\sigma _{n_{\rm{H}}}^{\rm{p}}\left( t \right)$, and $p_{{\hat h^{\rm{S}}_{n_{\rm{H}}}}|{h_{n_{\rm{H}}}^{\rm{S}}}}( {{\hat h_{n_{\rm{H}}}^{\rm{S}}}|{h_{n_{\rm{H}}}^{\rm{S}}}} )$ into (\ref{a28}) and (\ref{a29}) to obtain $\bar{h} _{{n_{\rm{H}}}}^{\rm{S}}\left( t \right)$ and $\sigma _{{n_{\rm{H}}}}^{{\rm{Hs}}}\left( t \right)$.
  \STATE 
    For each ${n_{\rm{H}}}$, plug $\hat h_{{n_{\rm{H}}}}^{\rm{S}}$, ${{\bar p}_{{n_{\rm{H}}}}\left( t \right)}$, and $\sigma _{{n_{\rm{H}}}}^{\rm{p}}\left( t \right)$ into (\ref{a32}) and (\ref{a33}) to obtain ${\bar s_{{n_{\rm{H}}}}}\left( t \right)$ and $\sigma _{{n_{\rm{H}}}}^{\rm{s}}\left( t \right)$.
  \STATE 
    For each ${n_{\rm{s}}}$, plug ${\hat x_{{n_{\rm{s}}}}}\left( t \right)$, ${\sigma _{{n_{\rm{s}}}}^{\rm{x}}\left( t \right)}$, ${{{\hat h}_{{n_{\rm{H}} },{n_{\rm{s}}}}}\left( t \right)}$, ${\sigma _{{n_{\rm{H}} },{n_{\rm{s}}}}^{\rm{h}}\left( t \right)}$, ${{{\bar s}_{n_{\rm{H}} }}\left( t \right)}$, and ${\sigma _{n_{\rm{H}} }^{\rm{s}}\left( t \right)}$ into (\ref{a54}) and (\ref{a55}) to obtain ${\bar r_{{n_{\rm{s}}}}}\left( t \right)$ and $\sigma _{{n_{\rm{s}}}}^{\rm{r}}\left( t \right)$.
  \STATE 
    For each ${n_{\rm{H}}}$ and ${n_{\rm{s}}}$, plug ${\hat x_{{n_{\rm{s}}}}}\left( t \right)$, ${\sigma _{{n_{\rm{s}}}}^{\rm{x}}\left( t \right)}$, ${{{\hat h}_{{n_{\rm{H}} },{n_{\rm{s}}}}}\left( t \right)}$, ${\sigma _{{n_{\rm{H}} },{n_{\rm{s}}}}^{\rm{h}}\left( t \right)}$, ${{{\bar s}_{n_{\rm{H}} }}\left( t \right)}$, and ${\sigma _{n_{\rm{H}} }^{\rm{s}}\left( t \right)}$ into (\ref{a56}) and (\ref{a57}) to obtain ${\bar q_{{n_{\rm{H}}},{n_{\rm{s}}}}}\left( t \right)$ and $\sigma _{{n_{\rm{H}}},{n_{\rm{s}}}}^{\rm{q}}\left( t \right)$.
  \STATE 
    For each ${n_{\rm{s}}}$, plug ${\bar r_{{n_{\rm{s}}}}}\left( t \right)$ and $\sigma _{{n_{\rm{s}}}}^{\rm{r}}\left( t \right)$ into (\ref{a42}) and (\ref{a43}) to obtain ${\hat x}_{n_{\rm{s}}}\left( t+1 \right)$ and $\sigma _{{n_{\rm{s}}}}^{\rm{x}}\left( t+1 \right)$.
  \STATE 
    For each ${n_{\rm{H}}}$ and ${n_{\rm{s}}}$, plug ${\bar q_{{n_{\rm{H}}},{n_{\rm{s}}}}}\left( t \right)$, $\sigma _{{n_{\rm{H}}},{n_{\rm{s}}}}^{\rm{q}}\left( t \right)$, and ${p_{{{\tilde h}_{{n_{\rm{H}}},{n_{\rm{s}}}}}|{\bm{x}}\left(t\right)}}( {{{\tilde h}_{{n_{\rm{H}}},{n_{\rm{s}}}}}|{\bm{x}}\left(t\right)} )$ into (\ref{a45}) and (\ref{a46}) to obtain ${\hat h_{{n_{\rm{H}}},{n_{\rm{s}}}}}\left( {t + 1} \right)$ and $\sigma _{{n_{\rm{H}}},{n_{\rm{s}}}}^{\rm{h}}\left( {t + 1} \right) $.
  \STATE 
    For each ${n_{\rm{H}}}$ and ${n_{\rm{s}}}$, plug ${\hat x_{{n_{\rm{s}}}}\left( t+1 \right)}$ into (\ref{a4}) to obtain the free space channel coefficient distribution with occlusion ${p_{{{\tilde h}_{{n_{\rm{H}}},{n_{\rm{s}}}}}|{\bm{\hat x}}\left(t+1\right)}}( {{{\tilde h}_{{n_{\rm{H}}},{n_{\rm{s}}}}}|{\bm{\hat x}}\left(t+1\right)} )$.
  \STATE
    If $\sum_{{n_{\rm{H}} }} {\left| {\hat h_{{n_{\rm{H}}}}^{\rm{S}}} - {\bar h_{{n_{\rm{H}}}}^{\rm{S}}}\left( {t} \right) \right|}  > \varepsilon_{\rm{t}} $, where $\varepsilon_{\rm{t}} $ is a given error tolerance value, stop the iteration.
  \ENDFOR
  \ENSURE
    Estimated sparse vector ${\hat x}_{n_{\rm{s}}}\left( t \right)$ and $\sigma _{{n_{\rm{s}}}}^{\rm{x}}\left( t \right)$.
  \end{algorithmic}
  \end{algorithm}

\subsection{Algorithm Iterative Execution}\label{it}
In Section \ref{yz2bl} and Section \ref{bl2yz}, we have obtained the estimation methods for variables ${\hat x_{{n_{\rm{s}}}}}\left( {t + 1} \right)$ and ${\hat h_{{n_{\rm{H}}},{n_{\rm{s}}}}}\left( {t + 1} \right)$, but the intermediate variables ${\bar p_{{n_{\rm{H}}}}}\left( t \right)$, ${{\bar r}_{{n_{\rm{s}}}}}\left( t \right)$, and ${\bar q_{{n_{\rm{H}}},{n_{\rm{s}}}}}\left( t \right)$ contain the passed messages ${\bar x_{{n_{\rm{s}}},{n_{\rm{H}}}}}\left( {t + 1} \right)$ and ${\bar h_{{n_{\rm{H}}}{n_{\rm{s}}},{n_{\rm{H}}}}}\left( {t + 1} \right)$. We need to replace the passed message with the current estimated variable ${\hat x_{{n_{\rm{s}}}}}\left( {t + 1} \right)$ and ${\hat h_{{n_{\rm{H}}},{n_{\rm{s}}}}}\left( {t + 1} \right)$ to calculate these intermediate variables again to achieve the iteration of the algorithm.

Pluging (\ref{a44}) and (\ref{a50}) into (\ref{a19}), we can get
{\small \begin{align}
  &{\bar p_{{n_{\rm{H}}}}}\left( t \right) \approx \sum\limits_{k = 1}^{{N_{\rm{s}}}} {{{\hat h}_{{n_{\rm{H}}},k}}\left( t \right){{\hat x}_k}\left( t \right)}  - {\bar s_{{n_{\rm{H}}}}}\left( {t - 1} \right) \nonumber \\ & \quad \times
  \sum\limits_{k = 1}^{{N_{\rm{s}}}} {\left( {\sigma _{{n_{\rm{H}}},k}^{\rm{h}}\left( t \right)\hat x_k^2\left( t \right) + \hat h_{{n_{\rm{H}}},k}^2\left( t \right)\sigma _k^{\rm{x}}\left( t \right)} \right)},\label{a52}
  \end{align}}
where we replace ${\hat x_k\left( t \right)}{\hat x_k\left( t-1 \right)}$ with ${\hat x_k^2\left( t \right)}$, replace ${\hat h_{{n_{\rm{H}}},k}\left( t \right)}{\hat h_{{n_{\rm{H}}},k}\left( t-1 \right)}$ with ${\hat h_{{n_{\rm{H}}},k}^2\left( t \right)}$, and drop the infinitesimal term under the CLT condition.
In most AMP algorithms \cite{Rangan,Jason}, (\ref{a52}) can be interpreted as the Onsager correction to ${{{\hat h}_{{n_{\rm{H}}},k}}\left( t \right){{\hat x}_k}\left( t \right)}$, and (\ref{a20}) is expressed as
{\small \begin{align}
  \sigma _{n_{\rm{H}}}^{\rm{p}}\left( t \right) \approx & \sum\limits_{k = 1}^{{N_{\rm{s}}}} {\left( {\sigma _{{n_{\rm{H}}},k}^{\rm{h}}\left( t \right)\hat x_k^2\left( t \right) }\right.}\nonumber\\ & {\left.{ + \hat h_{{n_{\rm{H}}},k}^2\left( t \right)\sigma _k^{\rm{x}}\left( t \right) + \sigma _{{n_{\rm{H}}},k}^{\rm{h}}\left( t \right)\sigma _k^{\rm{x}}\left( t \right)} \right)}.\label{a53}
\end{align}}

Similar to the approximate method in the (\ref{a52}), pluging (\ref{a44}) and (\ref{a50}) into (\ref{a40}) and (\ref{a41}), it yields
{\small \begin{align}
  {\bar r_{{n_{\rm{s}}}}}\left( t \right)  & \approx  {\hat x_{{n_{\rm{s}}}}}\left( t \right)\left( {1 - \sigma _{{n_{\rm{s}}}}^{\rm{r}}\left( t \right)\sum\limits_{k = 1}^{{N_{\rm{H}}}} {\sigma _{k,{n_{\rm{s}}}}^{\rm{h}}\left( t \right)\sigma _k^{\rm{s}}\left( t \right)} } \right)\nonumber \\ & + \sigma _{{n_{\rm{s}}}}^{\rm{r}}\left( t \right)\sum\limits_{k = 1}^{{N_{\rm{H}}}} {\left( {{{\hat h}_{k,{n_{\rm{s}}}}}\left( t \right){{\bar s}_k}\left( t \right)} \right)},\label{a54}
  \end{align}}
\begin{equation}\small
  \sigma _{{n_{\rm{s}}}}^{\rm{r}}\left( t \right) \approx {\left( {\sum\limits_{k = 1}^{{N_{\rm{H}}}} {\hat h_{k,{n_{\rm{s}}}}^2\left( t \right)\sigma _{{n_{\rm{s}}}}^{\rm{s}}\left( t \right)} } \right)^{ - 1}}.\label{a55}
  \end{equation}

And similar to the approximate method in the (\ref{a52}), pluging (\ref{a44}) and (\ref{a50}) into (\ref{a48}) and (\ref{a49}), it yields
{\small \begin{align}
  {\bar q_{{n_{\rm{H}}},{n_{\rm{s}}}}}\left( t \right) & \approx {\hat h_{{n_{\rm{H}}},{n_{\rm{s}}}}}\left( t \right)\left( {1 - \sigma _{{n_{\rm{H}}},{n_{\rm{s}}}}^{\rm{q}}\left( t \right)\sigma _{{n_{\rm{s}}}}^{\rm{x}}\left( t \right)\sigma _{{n_{\rm{H}}}}^{\rm{s}}\left( t \right)} \right) \nonumber\\ & + \sigma _{{n_{\rm{H}}},{n_{\rm{s}}}}^{\rm{q}}\left( t \right){\hat x_{{n_{\rm{s}}}}}\left( t \right){\bar s_{{n_{\rm{H}}}}}\left( t \right),\label{a56}
\end{align}}
\begin{equation}\small
  \sigma _{{n_{\rm{H}}},{n_{\rm{s}}}}^{\rm{q}}\left( t \right) \approx {\left( {\hat x_{{n_{\rm{s}}}}^2\left( t \right)\sigma _{{n_{\rm{H}}}}^{\rm{s}}\left( t \right)} \right)^{ - 1}}.\label{a57}
  \end{equation}

We summarize the proposed GAMP-MVSVR algorithm in Algorithm \ref{sf1}.
During the execution process of the GAMP-MVSVR algorithm, the approximate message passing and the occlusion detection method execute repeatedly and converge to the estimated value of environmental information.
The computational complexity of the GAMP-MVSVR algorithm consists of two parts: approximate message passing and occlusion detection method.
The computational complexity is expressed as $\mathcal{O}\left(N_{\rm{H}}N_{\rm{s}}+N_{\rm{H}}N_{\rm{s}}\left\lVert \bm{x}\right\rVert_0 \right)$, where $\left\lVert \bm{x}\right\rVert_0$ is the number of the environmental scatterers, the number of pixels $N_{\rm{s}}$ has a major influence on the computational complexity. 
The setting of $T_{\rm max}$ is to stop the iteration in time when the system performance is poor. Due to the small order of magnitude of iterations, the number of algorithm iterations will hardly reach $T_{\rm max}$, and we omit the term $T_{\rm max}$ when calculating the computational complexity.

\section{System Performance Analysis}\label{fx}
In this section, we compare the performance of the single-BS environment sensing scheme and the multi-BS joint environment sensing scheme from two aspects: environment sensing range and environment sensing accuracy.

\subsection{Environment Sensing Range}
In this section, we analyze the impact of the CS measurement matrix ${{\bf{\tilde H}}}$ on the system performance.
According to the CS reconstruction model in Section \ref{ys}, the environment sensing range is determined by the CS measurement matrix ${{\bf{\tilde H}}}$. Specifically, the number of 0 elements in the occlusion matrix ${{\bf{V}}}$, ${{\bf{V}}^{{\rm{U\rightarrow s}}}}$, and ${{\bf{V}}^{{\rm{s\rightarrow B}}}}$ is calculated to obtain the environment sensing range.

As mentioned in Section \ref{mx}, the number of pixels is $N_{\rm{s}}$, and the position of ${n_{\rm{s}}}$-th pixel is ${\bm{a}}_{n_{\rm{s}}}^{\rm{s}} = \left[x_{n_{\rm{s}}}^{\rm{s}},y_{n_{\rm{s}}}^{\rm{s}},z_{n_{\rm{s}}}^{\rm{s}}\right]$.
The number of users in the environment is $N_{\rm{u}}$, and the position of $n_{\rm{u}}$-th user is ${\bm{a}}_{n_{\rm{u}}}^{\rm{u}} = \left[x_{n_{\rm{u}}}^{\rm{u}},y_{n_{\rm{u}}}^{\rm{u}},z_{n_{\rm{u}}}^{\rm{u}}\right]$.
In the multi-BS joint environment sensing scheme, the number of BSs in the environment is $N_{\rm{B}}$, and the position of $n_{\rm{B}}$-th BS is ${\bm{a}}_{n_{\rm{B}}}^{\rm{B}} = \left[x_{n_{\rm{B}}}^{\rm{B}},y_{n_{\rm{B}}}^{\rm{B}},z_{n_{\rm{B}}}^{\rm{B}}\right]$. 
Under these definitions, the analysis of the occlusion matrix ${{\bf{V}}^{{\rm{U\rightarrow s}}}}$ and ${{\bf{V}}^{{\rm{s\rightarrow B}}}}$ is summarized in the following theorem.
\begin{theorem}
  The probability that the element in the occlusion matrix ${{\bf{V}}^{{\rm{U\rightarrow s}}}}$ is 0 can be expressed as
\begin{equation}\small
  \begin{split}
  &p\left( {{\bf{V}}_{{n_{\rm{u}}},{n_{\rm{s}}}}^{{\rm{U\rightarrow s}}} = 0} \right) = 1-\frac{{ {{\left\| {{{\bf{V}}^{{\rm{U\rightarrow s}}}}} \right\|}_1}}}{{{N_{\rm{u}}} \times {N_{\rm{s}}}}}\\ &\quad \quad \quad \quad \quad\quad = \int_{{\bm{a}}_{{n_{\rm{u}}}}^{\rm{u}},{\bm{a}}_{{n_{\rm{s}}}}^{\rm{s}}} {\frac{{{p_{\rm{u}}}\left( {{\bm{a}}_{{n_{\rm{u}}}}^{\rm{u}}} \right)\lambda l^2{{\left\| {{\bm{a}}_{{n_{\rm{u}}}}^{\rm{u}} - {\bm{a}}_{{n_{\rm{s}}}}^{\rm{s}}} \right\|}_2}}}{{{{\left( {{L_{\rm{s}}}{W_{\rm{s}}}{H_{\rm{s}}}} \right)}^2}}}}, \label{s1}
  \end{split}
\end{equation}
where ${{p_{\rm{u}}}\left( {{\bm{a}}_{{n_{\rm{u}}}}^{\rm{u}}} \right)}$ is the distribution function of $n_{\rm{u}}$-th user in the environment. For the receiving antenna deployed on the $n_{\rm{B}}$-th BS, the probability that the element in the occlusion matrix ${{\bf{V}}^{{\rm{s\rightarrow B}}}}\left( {{n_{\rm{R}}}} \right)$ is 0 can be expressed as 
\begin{equation}\small
  \begin{split}
  &p\left( {{\bf{V}}_{{n_{\rm{s}}}}^{{\rm{s\rightarrow B}}}\left( {{n_{\rm{R}}}} \right) = 0} \right) = 1-\frac{{ {{\left\| {{{\bf{V}}^{{\rm{s\rightarrow B}}}}} \right\|}_1}}}{{{N_{\rm{s}}}}}\\ &\quad \quad \quad \quad \quad = \int_{{\bm{a}}_{{n_{\rm{B}}}}^{\rm{B}},{\bm{a}}_{{n_{\rm{s}}}}^{\rm{s}}} {\frac{{{p_{\rm{B}}}\left( {{\bm{a}}_{{n_{\rm{B}}}}^{\rm{B}}} \right)\lambda l^2{{\left\| {{\bm{a}}_{{n_{\rm{B}}}}^{\rm{B}} - {\bm{a}}_{{n_{\rm{s}}}}^{\rm{s}}} \right\|}_2}}}{{{{\left( {{L_{\rm{s}}}{W_{\rm{s}}}{H_{\rm{s}}}} \right)}^2}}}}, \label{s3}
  \end{split}
\end{equation}
where ${{p_{\rm{B}}}\left( {{\bm{a}}_{{n_{\rm{B}}}}^{\rm{B}}} \right)}$ is the distribution function of $n_{\rm{B}}$-th BS in the environment.
  \label{th1}
  \end{theorem}

\begin{proof}
  According to the occlusion detection model in Section \ref{zd}, when there is a scatterer in the space of size $l^2{{\left\| {{\bm{a}}_{{n_{\rm{u}}}}^{\rm{u}} - {\bm{a}}_{{n_{\rm{s}}}}^{\rm{s}}} \right\|}_2}$ between the $n_{\rm{u}}$-th user and the ${n_{\rm{s}}}$-th pixel, the EM propagation path is blocked. 
  The location distribution function ${{p_{\rm{u}}}\left( {{\bm{a}}_{{n_{\rm{u}}}}^{\rm{u}}} \right)}$ of each user is independent, and the division of pixels is uniform.
  Therefore, in the space of size  ${{L_{\rm{s}}}{W_{\rm{s}}}{H_{\rm{s}}}}$, the probability of a scatterer between the $n_{\rm{u}}$-th user and the ${n_{\rm{s}}}$-th pixel is denoted as (\ref{s1}).
  In the occlusion detection model, the analysis methods of the BS and the user are the same. Therefore, according to (\ref{s1}), replacing ${{p_{\rm{u}}}\left( {{\bm{a}}_{{n_{\rm{u}}}}^{\rm{u}}} \right)}$ with ${{p_{\rm{B}}}\left( {{\bm{a}}_{{n_{\rm{B}}}}^{\rm{B}}} \right)}$ and replacing ${\bm{a}}_{n_{\rm{s}}}^{\rm{s}}$ with ${\bm{a}}_{n_{\rm{B}}}^{\rm{B}}$, we obtain (\ref{s3}).
\end{proof}

The occlusion matrix ${{\bf{V}}^{{\rm{U\rightarrow s}}}}$ reflects the different views from different users. 
In the single-BS environment sensing scheme and the multi-BS joint environment sensing scheme, the multi-views from users are the same.
As mentioned in Section \ref{mx}, when ${{\bf{V}}^{{\rm{U\rightarrow s}}}}$ contains an all-zero column, it means that all users cannot sense the corresponding pixel, and the pixel is out of the user sensing range.
According to Theorem \ref{th1}, when the number of users is ${N_{\rm{u}}}$, the probability that the pixels in the environment are outside the user sensing range is
\begin{equation}\small
  p\left( {{\bf{V}}_{:,{n_{\rm{s}}}}^{{\rm{U\rightarrow s}}} = \bf{0}} \right) = p^{N_{\rm{u}}}\left( {{\bf{V}}_{{n_{\rm{u}}},{n_{\rm{s}}}}^{{\rm{U\rightarrow s}}} = 0} \right). \label{s2}
\end{equation}

Let ${{\bf{V}}^{{\rm{s\rightarrow B}}}} = \left[{{\bf{V}}^{{\rm{s\rightarrow B}}}}\left( 1 \right),{{\bf{V}}^{{\rm{s\rightarrow B}}}}\left( 2 \right),\ldots,{{\bf{V}}^{{\rm{s\rightarrow B}}}}\left( {{N_{\rm{R}}}} \right)\right] \in \left\{0,1\right\} ^{N_{\rm{s}} \times N_{\rm{R}}}$ and the occlusion matrix ${{\bf{V}}^{{\rm{s\rightarrow B}}}}$ reflects the different views from different BSs.
In the single-BS environment sensing scheme, since the receiving antennas are deployed close to each other, they have almost the same views to the scatterers in the environment. Therefore, the probability that a pixel is outside the sensing range of a single BS is $p\left( {{\bf{V}}_{{n_{\rm{s}}}}^{{\rm{s\rightarrow B}}}\left( {{n_{\rm{R}}}} \right) = 0} \right)$.
In the multi-BS joint environment sensing scheme, when ${{\bf{V}}^{{\rm{s\rightarrow B}}}}$ contains an all-zero row, it means that all BSs cannot sense the corresponding pixel, and the pixel is out of the BS sensing range.
According to Theorem \ref{th1}, since the location distribution probability of each BS is independent, the probability that a pixel in the environment is outside the sensing range of the BS is
\begin{equation}\small
  p\left( {{\bf{V}}_{{n_{\rm{s}}},:}^{{\rm{s\rightarrow B}}} = \bf{0}} \right) = p^{N_{\rm{B}}}\left( {{\bf{V}}_{{n_{\rm{s}}}}^{{\rm{s\rightarrow B}}}\left( {{n_{\rm{R}}}} \right) = 0} \right). \label{s4}
\end{equation}

According to (\ref{s1}), (\ref{s3}), (\ref{s2}), and (\ref{s4}), we calculate the mean value $\bar{N}_{\rm{con}}$ of the number of pixels that cannot be sensed in the single-BS environment sensing scheme and the mean value $\bar{N}_{\rm{dis}}$ of the number of pixels that cannot be sensed in the multi-BS joint environment sensing scheme as
{\small \begin{align}
  \bar{N}_{\rm{con}}& =N_{\rm{s}}\left({ p^{N_{\rm{u}}}\left( {{\bf{V}}_{{n_{\rm{u}}},{n_{\rm{s}}}}^{{\rm{U\rightarrow s}}} = 0} \right)+ p\left( {{\bf{V}}_{{n_{\rm{s}}}}^{{\rm{s\rightarrow B}}}\left( {{n_{\rm{R}}}} \right) = 0} \right)} \right. \nonumber\\ & \left. {- p^{N_{\rm{u}}}\left( {{\bf{V}}_{{n_{\rm{u}}},{n_{\rm{s}}}}^{{\rm{U\rightarrow s}}} = 0} \right)p\left( {{\bf{V}}_{{n_{\rm{s}}}}^{{\rm{s\rightarrow B}}}\left( {{n_{\rm{R}}}} \right) = 0} \right)} \right),\label{s6}
\end{align}}
{\small \begin{align}
  \bar{N}_{\rm{dis}} & = N_{\rm{s}}\left({ p^{N_{\rm{u}}}\left( {{\bf{V}}_{{n_{\rm{u}}},{n_{\rm{s}}}}^{{\rm{U\rightarrow s}}} = 0} \right)+ p^{N_{\rm{B}}}\left( {{\bf{V}}_{{n_{\rm{s}}}}^{{\rm{s\rightarrow B}}}\left( {{n_{\rm{R}}}} \right) = 0} \right)} \right. \nonumber \\ - & p^{N_{\rm{u}}}\left( {{\bf{V}}_{{n_{\rm{u}}},{n_{\rm{s}}}}^{{\rm{U\rightarrow s}}} = 0} \right)p^{N_{\rm{B}}}\left( {{\bf{V}}_{{n_{\rm{s}}}}^{{\rm{s\rightarrow B}}}\left( {{n_{\rm{R}}}} \right) = 0} \right) \Big).\label{s7}
\end{align}}

The multi-BS joint environment sensing has a larger sensing range than single-BS environment sensing. In practical applications, a multi-BS scheme should be implemented according to the BS location. 
When the multi-BS scheme is difficult to implement, e.g., when the BSs are far away and users cannot communicate with each BS, a single-BS scheme should be used. In addition, if there is a failure of environment sensing for some pixels, a larger number of base stations and users or users after moving locations should be used for joint sensing.

\subsection{Environment Sensing Accuracy} \label{gzzq}
In this section, we analyze the impact of environmental information $\bm{x}$ on system performance. The mean square error (MSE) between the estimated environmental information $\bm{\hat x}$ and the real environmental information $\bm{x}$ is used to describe the accuracy of environment sensing, i.e., 
\begin{equation}\small
  {\rm{MSE}} = \frac{1}{N_{\rm{s}}}\left\| {{\bm{x}} - {\bm{\hat x}}} \right\|_2^2.\label{s9}
\end{equation}

According to the CS reconstruction model in Section \ref{ys} (\ref{q1}), the reconstruction problem in (\ref{q2}) is expressed as
\begin{equation}\small
  {\bm{\hat x}} = \arg \mathop {\min }\limits_{\bm{x}} {\left\| {\bm{x}} \right\|_1}\quad \quad {\rm{s}}{\rm{.t}}{\rm{.}}\quad {\left\| {{\bf{H}}^{\rm{S}}} - {\bf{\tilde H}}{\bm x} \right\|_2} \le \varepsilon ,\label{s10}
\end{equation}
where $\varepsilon$ is a slack variable.
We consider that the performance of environment sensing methods for solving the occlusion problem should be as close as possible to the upper limit in Theorem \ref{th2}.

\begin{figure*}[t]
  \centering
  \begin{minipage}[t]{0.48\textwidth}
      \centering
      \subfigure[\footnotesize The original environment scatterer distribution.]{
      \includegraphics[width=4cm]{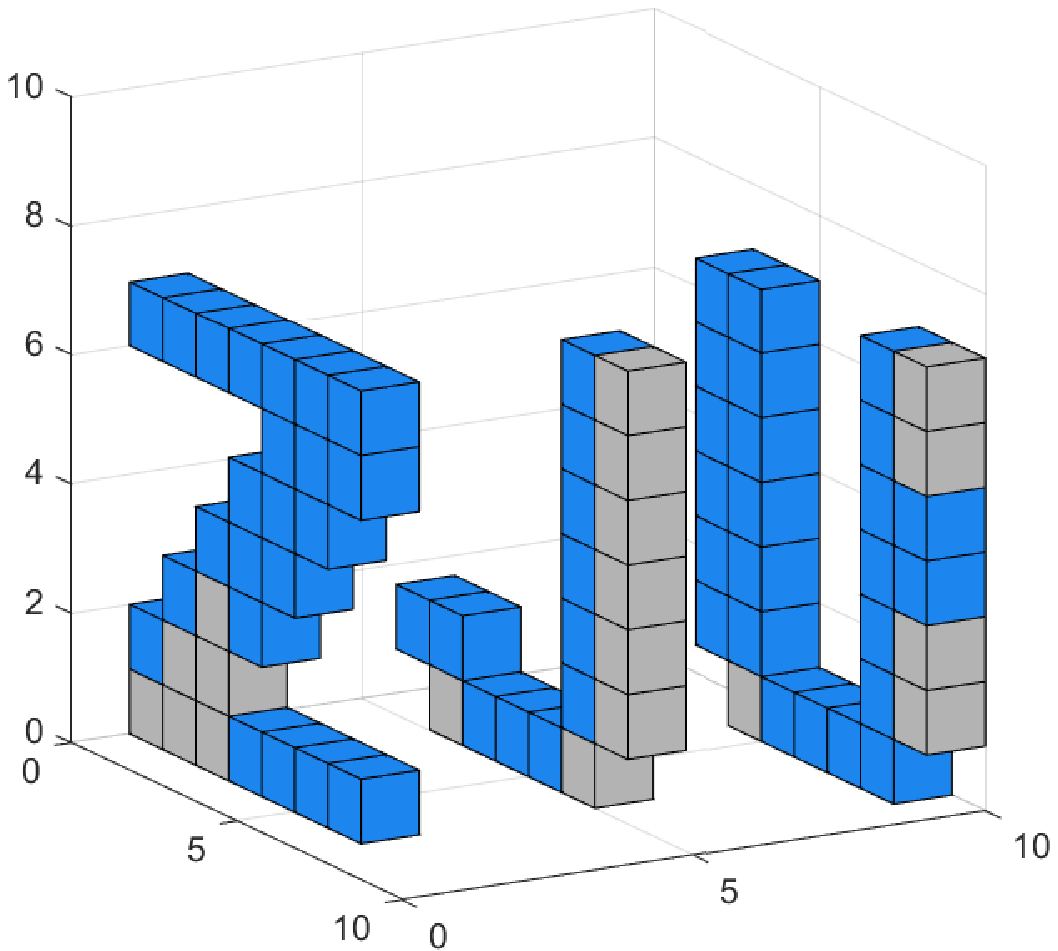}}
      \subfigure[\footnotesize The result of the GAMP algorithm.]{
      \includegraphics[width=4cm]{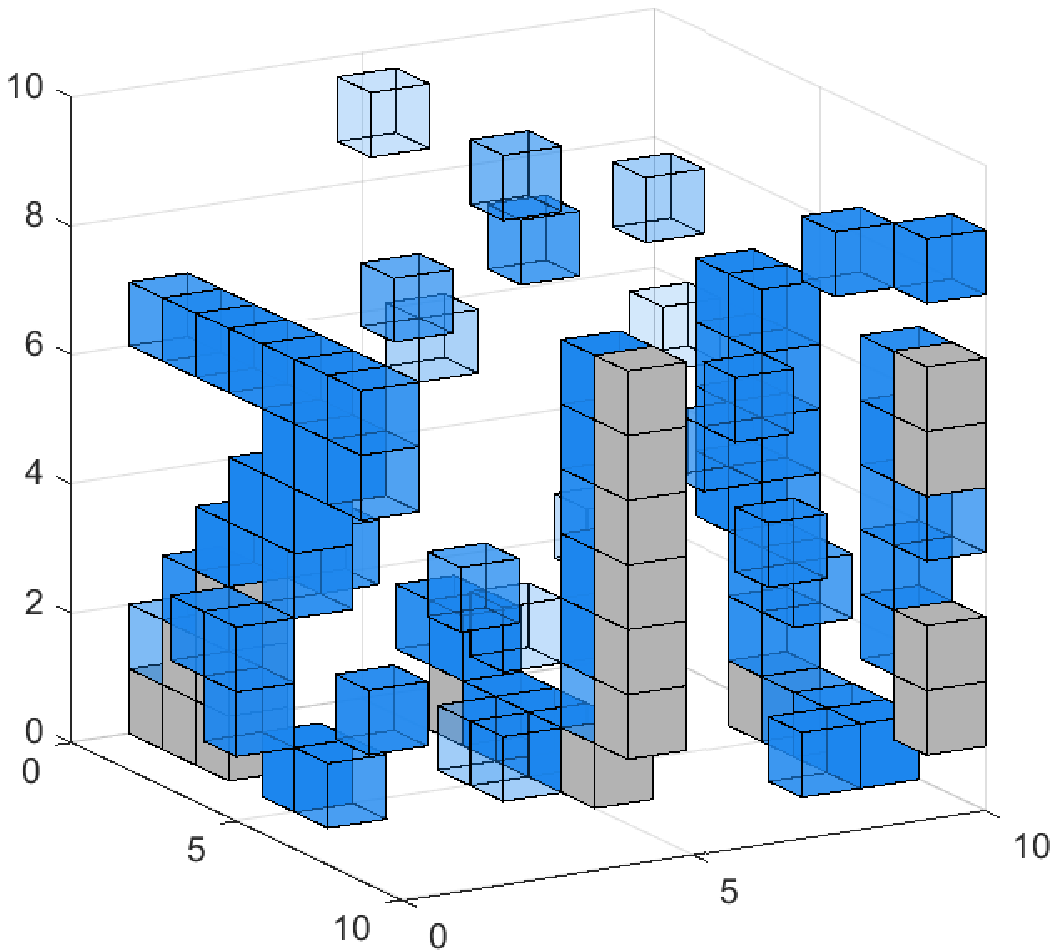}}
    
      \subfigure[\footnotesize The result of the Bilinear GAMP algorithm.]{
      \includegraphics[width=4cm]{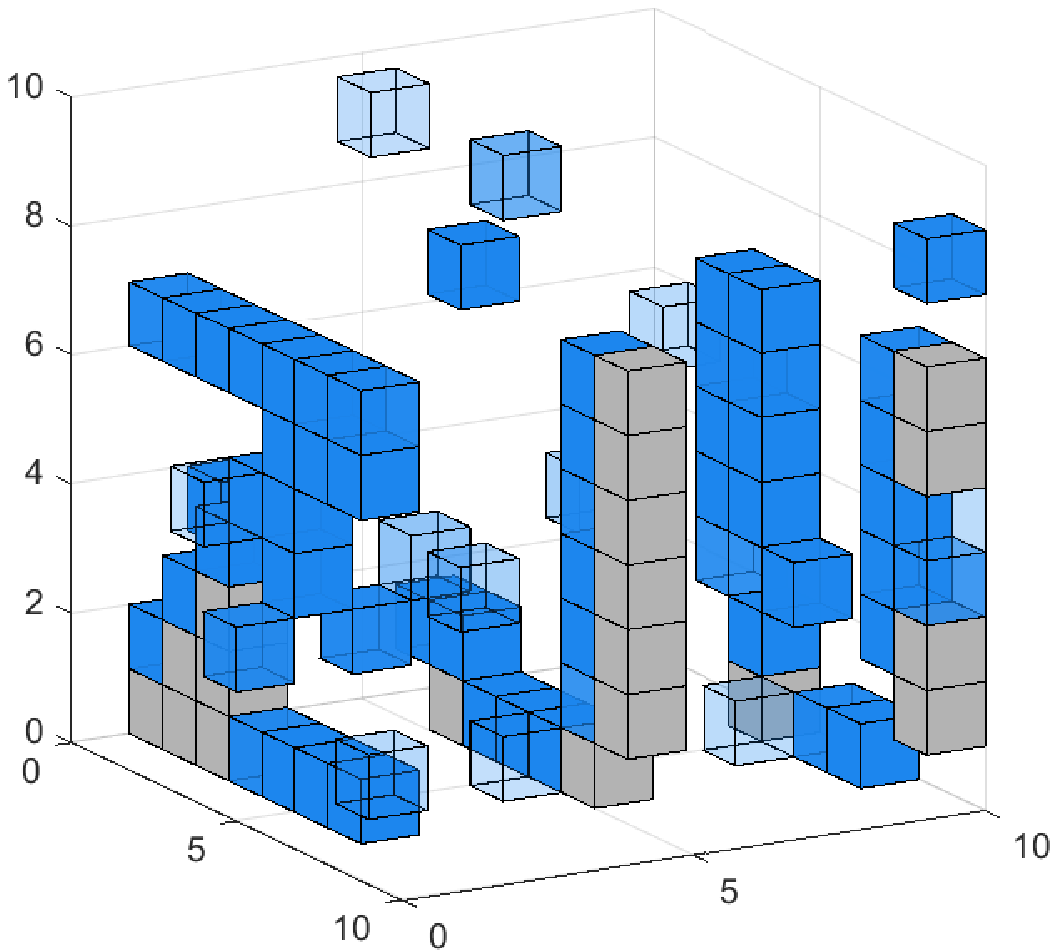}}
      \subfigure[\footnotesize The result of the proposed GAMP-MVSVR algorithm.]{
      \includegraphics[width=4cm]{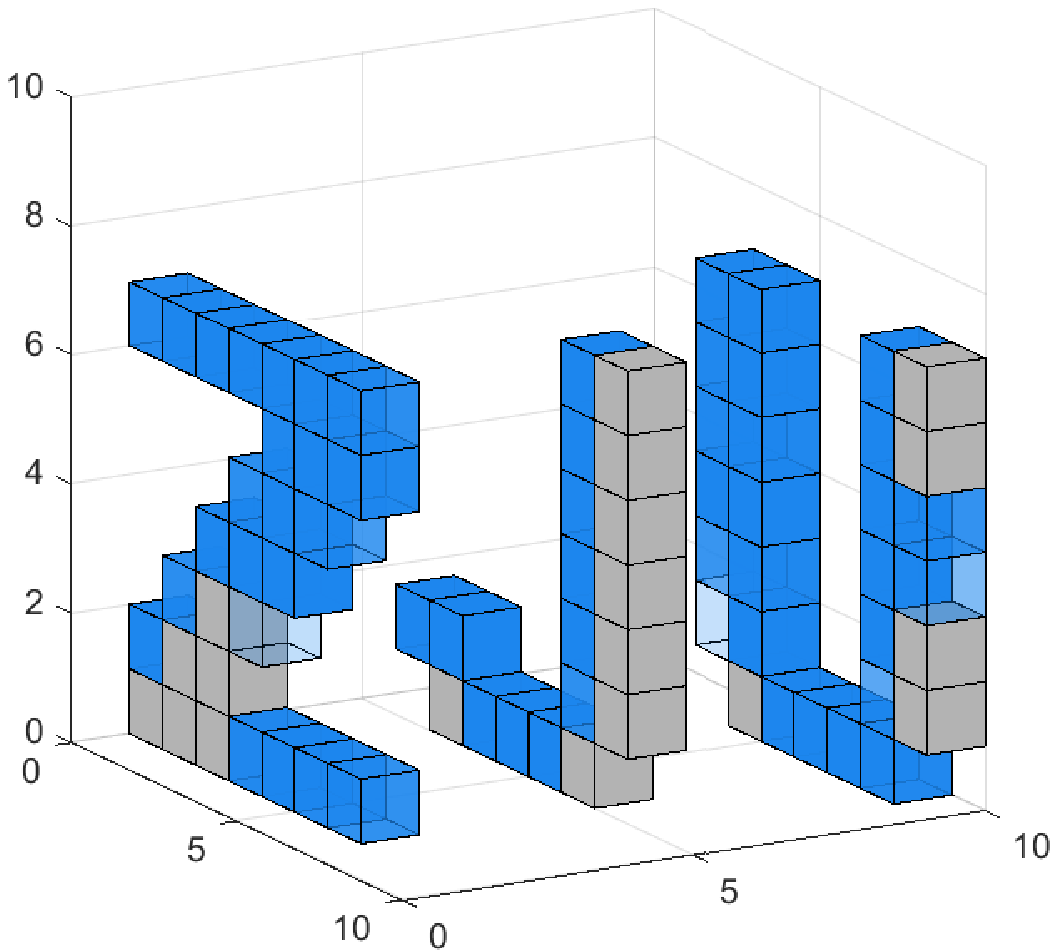}}
      \caption{The original environment scatterer distribution and single-BS environment sensing results.}
      \label{figa1}
  \end{minipage}
  \begin{minipage}[t]{0.48\textwidth}
      \centering
      \subfigure[\footnotesize The original environment scatterer distribution.]{
      \includegraphics[width=4cm]{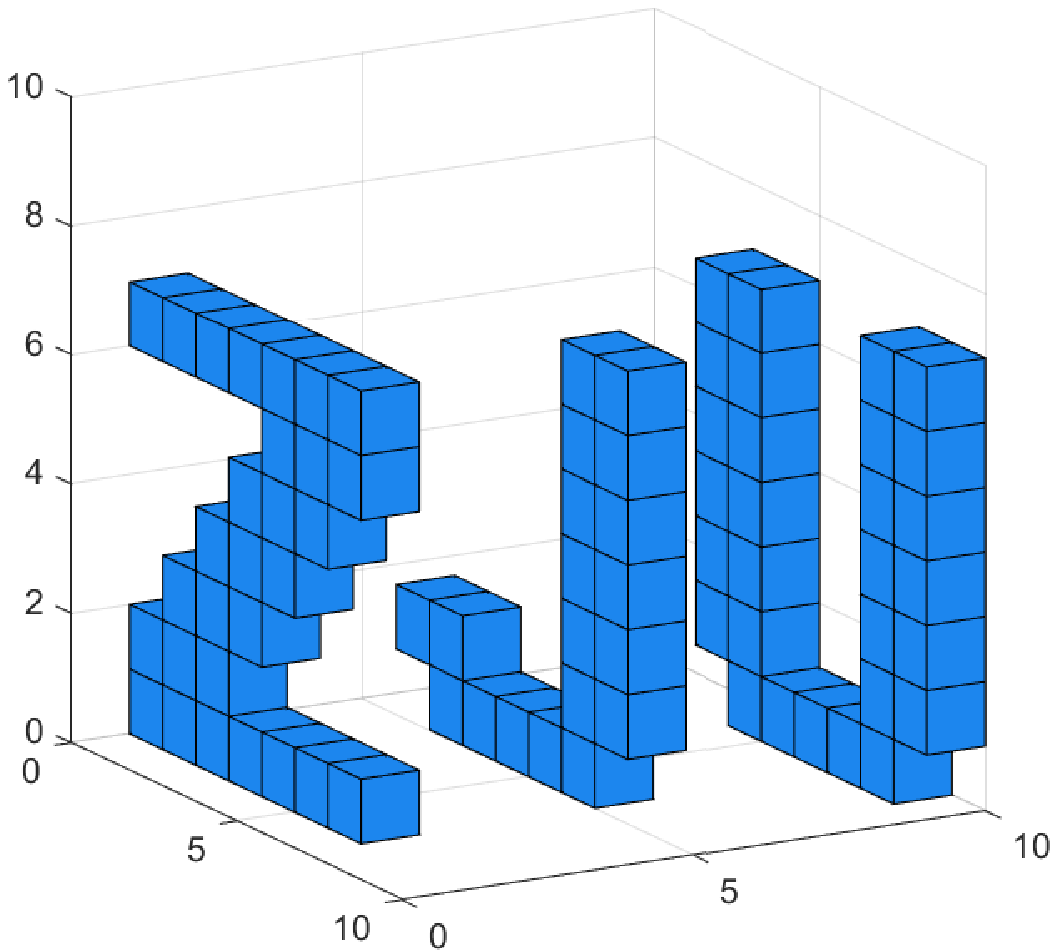}}
      \subfigure[\footnotesize The result of the GAMP algorithm.]{
      \includegraphics[width=4cm]{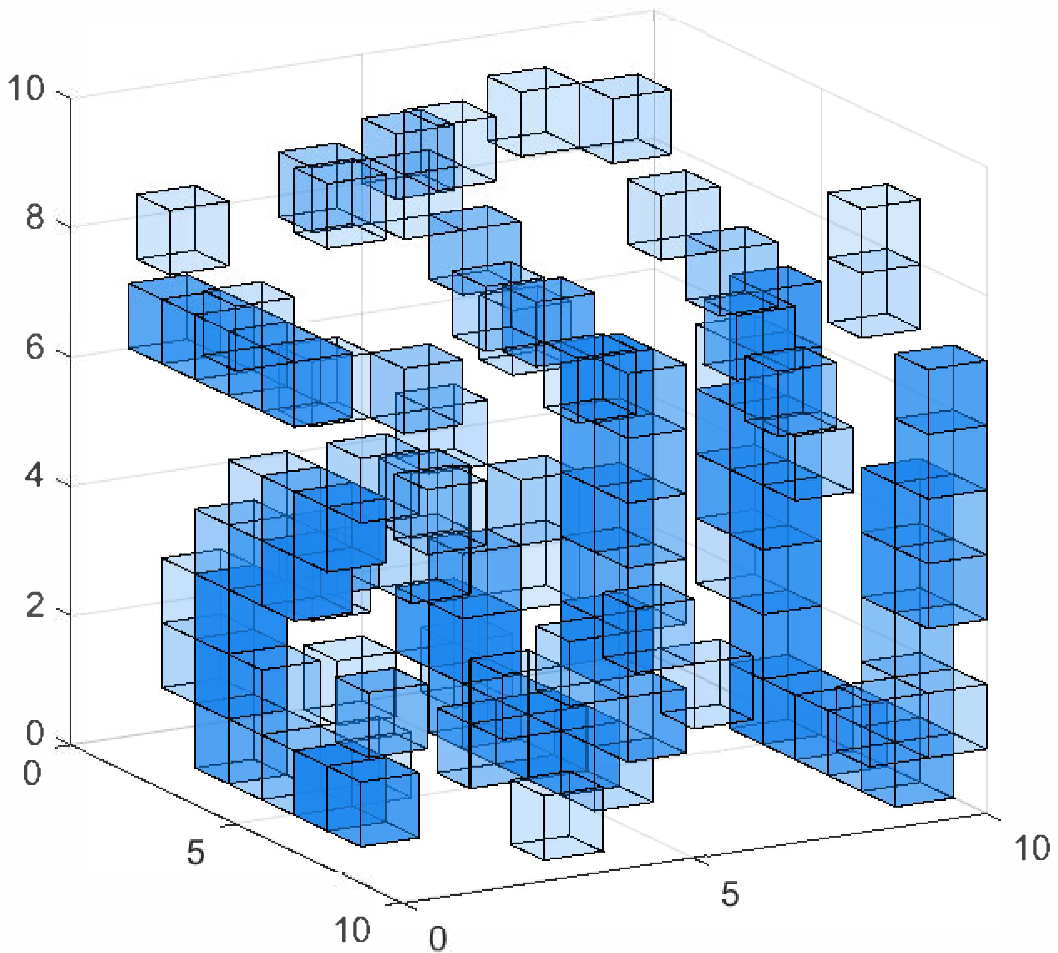}}
    
      \subfigure[\footnotesize The result of the Bilinear GAMP algorithm.]{
      \includegraphics[width=4cm]{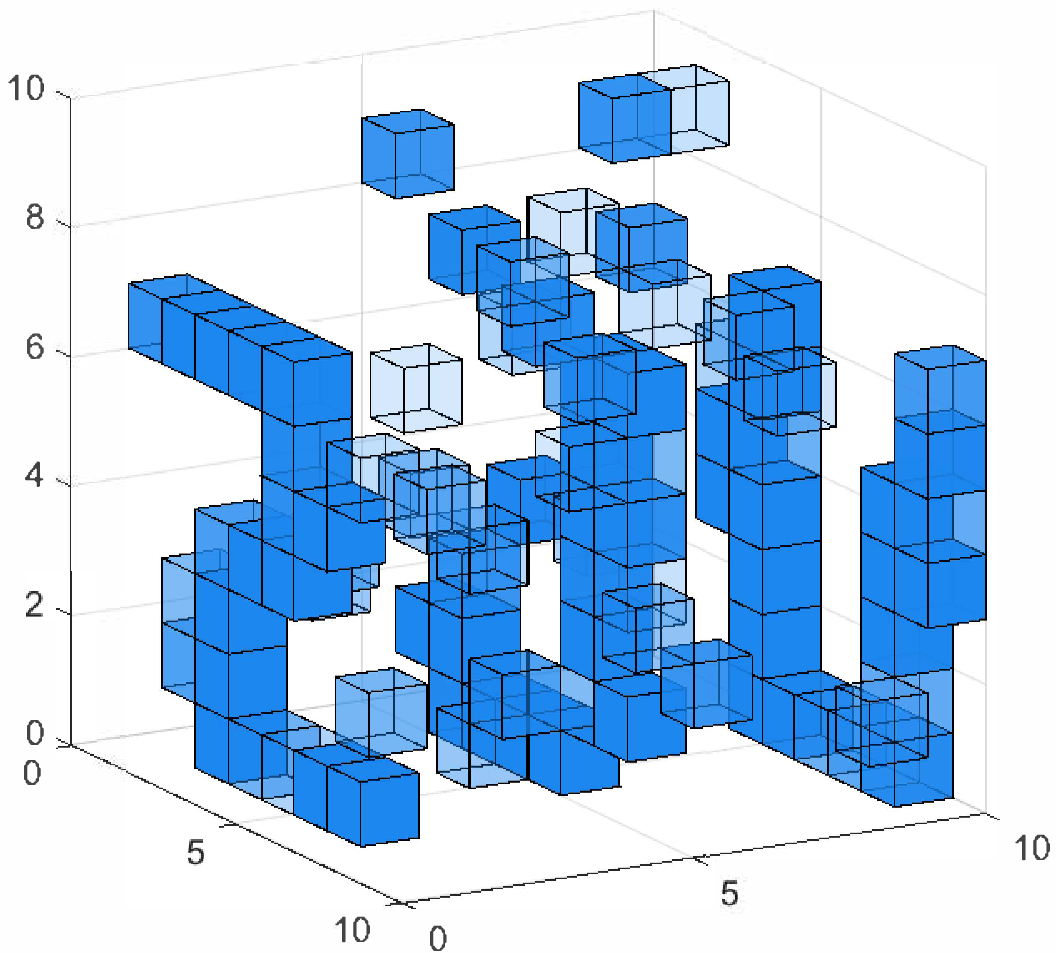}}
      \subfigure[\footnotesize The result of the proposed GAMP-MVSVR algorithm.]{
      \includegraphics[width=4cm]{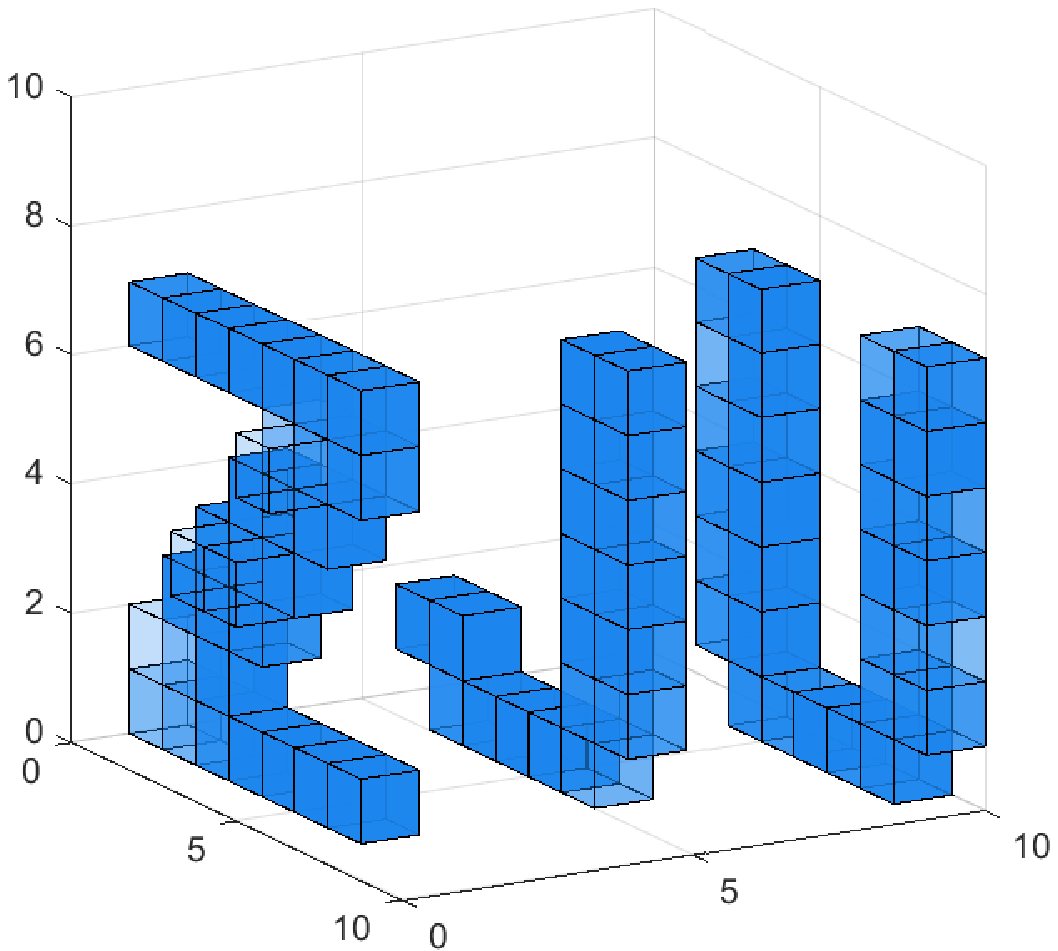}}
      \caption{The original environment scatterer distribution and multi-BS joint environment sensing results.}
      \label{figa2}
  \end{minipage}
\end{figure*}

\begin{theorem}
The upper bound of the environment sensing error when the CS measurement matrix ${{\bf{\tilde H}}}={\bf{H}}$ under the ideal (no occlusion, the occlusion detection is completely correct) situation is
\begin{equation}\small
  {\rm{MSE}} \le \frac{ m \cdot R^{2} \cdot \log {N_{\rm{s}}}}{N_{\rm{s}} \cdot N_{\rm{u}} \cdot N_{\rm{R}}} ,\label{s5}
\end{equation}
where $m$ is a normalized constant, and
{\small \begin{align}
    R = \left\{ \begin{array}{l}
   \! R_{\rm{con}} = \left\lVert \bm{x}_{\rm{con}}\right\rVert_1 = \lambda \theta^{\rm{x}}(N_{\rm{s}} - {\bar N}_{\rm{con}}), \\
   \quad\quad\quad\quad\quad\quad\quad\quad\quad\text{for single-BS scheme},\\
   \! R_{\rm{dis}} = \left\lVert \bm{x}_{\rm{dis}}\right\rVert_1 = \lambda \theta^{\rm{x}} (N_{\rm{s}} - {\bar N}_{\rm{dis}}), \\
   \quad\quad\quad\quad\quad\quad\quad\quad\quad\text{for multi-BS scheme}.
    \end{array} \right.
\end{align}}

\label{th2}
\end{theorem}

\begin{proof}
According to the CS theory \cite{Donoho},
\begin{equation}\small
  \left\lVert {\bm{x}} - {\bm{\hat x}}\right\rVert_2 \leq m \cdot R \cdot \left(N_{\rm{H}}/ \log{N_{\rm{s}}}\right)^{1/2-1/{p_{\rm{s}}}},
  \label{s8}
\end{equation}
where, $\left\lVert \bm{x}\right\rVert_{p_{\rm{s}}} \leq R$. Pluging (\ref{s9}) and (\ref{s10}) into (\ref{s8}), we obtain (\ref{s5}). According to (\ref{s10}), let ${p_{\rm{s}}} = 1$, we get the values of $R_{\rm{con}}$ and $R_{\rm{dis}}$.
\end{proof}

\begin{remark}
 When the number of BS antennas or the number of users increases, the upper bound of environment sensing error MSE decreases.
When the environment sparsity $\lambda$ and other system parameters remain constant, the number of pixels $N_{\rm{s}}$ increases (this means a wider sensing range or higher sensing resolution), the upper bound of environment sensing error MSE increases. Therefore, wider and more accurate sensing results require more resource costs.
\end{remark}

According to (\ref{s6}) and (\ref{s7}), under the condition that the number of users is the same and the total number of antennas of all BSs is the same, the number of environmental scatterers sensed by multiple BSs is more than that of a single BS, $R_{\rm{dis}} > R_{\rm{con}}$.
Therefore, according to (\ref{s5}), for the scatterers in the sensing range, the upper bound of MSE of multi-BS joint environment sensing is higher than that of the single-BS scheme.
However, in practical applications, the total number of antennas deployed by multiple BSs is usually more than that of a single BS. The increase in the number of antennas $N_{\rm{R}}$ will significantly improve system performance.

We analyze the probability distribution of the occlusion matrix caused by the environmental scatterers as shown in Corollary \ref{th3}.

\begin{corollary}
  The elements in the occlusion matrix $\bf V$ conform to the Bernoulli distribution. The probability that the element in the occlusion matrix ${{\bf{V}}}_{{n_{\rm{u}}},{n_{\rm{s}}}}\left(N_{\rm R}\right)$ is 0 can be expressed as
  {\small
  \begin{align} 
      &p\left( {{\bf{V}}_{{n_{\rm{u}}},{n_{\rm{s}}}}\left(N_{\rm R}\right) = 0} \right) \nonumber \\ =& p\left( {{\bf{V}}_{{n_{\rm{u}}},{n_{\rm{s}}}}^{{\rm{U\rightarrow s}}} = 0} \right) + p\left( {{\bf{V}}_{{n_{\rm{s}}}}^{{\rm{s\rightarrow B}}}\left( {{n_{\rm{R}}}} \right) = 0} \right) \nonumber \\ -& \frac{p\left( {{\bf{V}}_{{n_{\rm{u}}},{n_{\rm{s}}}}^{{\rm{U\rightarrow s}}} = 0} \right)p\left( {{\bf{V}}_{{n_{\rm{s}}}}^{{\rm{s\rightarrow B}}}\left( {{n_{\rm{R}}}} \right) = 0} \right)}{N_{\rm s}}
  \end{align}} 
  \label{th3} 
\end{corollary}

\begin{proof}
  According to ${\bf{V}}\left( {{n_{\rm{R}}}} \right) = {{\bf{V}}^{{\rm{U\rightarrow s}}}}{\rm{diag}}\left( {{{\bf{V}}^{{\rm{s\rightarrow B}}}}\left( {{n_{\rm{R}}}} \right)} \right)$, as shown in Theorem \ref{th1}, the probability that the element in the occlusion matrix ${{\bf{V}}}_{{n_{\rm{u}}},{n_{\rm{s} }}}\left(N_{\rm R}\right)$ is 0 can be calculated by ${{\bf{V}}^{{\rm{U\rightarrow s}}}}$ and ${{{\bf{V}}^{{\rm{s\rightarrow B}}}}\left( {{n_{\rm{R}}}} \right)}$.
\end{proof}

Therefore, the unknown occlusion matrix $\bf V$ is equivalent to causing a multiplicative Bernoulli perturbation to the known free space channel matrix $\bf H$ as shown in (\ref{ww1}).
At present, some works \cite{Lavrenko,Yang} have analyzed the compressed sensing reconstruction performance under the condition that there is a specific perturbation in the compressed sensing measurement matrix.
In this paper, the proposed algorithm iteratively eliminates the perturbation to the measurement matrix caused by the occlusion effect according to the results of each iteration and obtains accurate imaging results.
The system performance analysis of this iterative detection process will be future work, and extensive simulation results have verified the effectiveness and convergence of the proposed algorithm.

\section{Numerical Results}
In this section, we simulate the performance of the algorithm. In order to reduce the computational cost and reflect the environment sensing accuracy, we set up a small-scale simulation sensing region, the environmental scatterers are distributed in the sensing region of $5\rm{m} \times 5\rm{m} \times 5\rm{m}$, and the sensing region is evenly divided into $10 \times 10 \times 10$ pixels.
Multiple BSs and users are randomly distributed in or around the sensing region.
The carrier frequency of the uplink signal is set to $30\rm{GHz}$. The positions of scatterers distributed in sensing region are random, and the scattering coefficient is set to $x_{n_{\rm{s}}} \in \left[0,1\right]$.
As shown in the Fig. \ref{figa1} and Fig. \ref{figa2}, small cubes are used to represent the position and scattering coefficient of pixels. The lower the transparency of the small cubes, the higher the scattering coefficient of the pixel.

In this section, as mentioned in Section \ref{GAMP2}, we choose GAMP and Bilinear GAMP algorithm as the baseline algorithms to solve the environment sensing problem under the occlusion effect.
Compared with the two baseline algorithms, in this paper, the proposed GAMP-MVSVR algorithm based on the occlusion effect significantly improves the accuracy of environment sensing.

  \begin{figure}[t]
    \centering
  \includegraphics[width=3.5in]{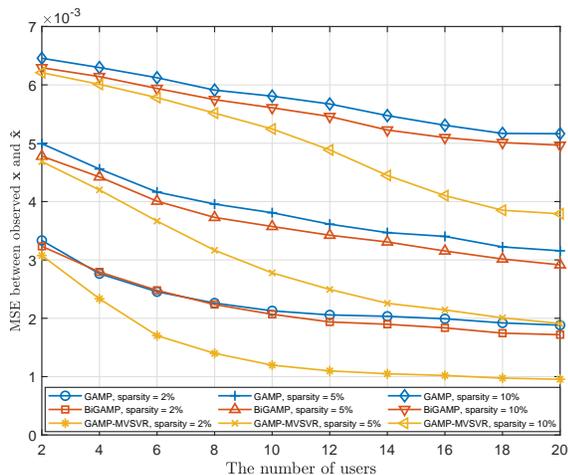}
  \caption{The relationship between the number of users and the MSE within the sensing range of the single-BS scheme.}
  \label{figa3}
  \end{figure}
  
Fig. \ref{figa1} shows the results of the single-BS environment sensing, where the gray cubes represent scatterers that are out of the sensing range due to the occlusion effect. 
The system parameters are set to the number of users $N_{\rm{u}}=20$, a single BS is deployed with a $5 \times 5$ array antenna, and the signal-to-noise ratio (SNR) $E_{\rm{b}}/N_{\rm{ 0}}=20\rm{dB}$.
Compared with the baseline algorithms, it can be seen that the proposed algorithm clearly and accurately obtains the shape of the original target.

Fig. \ref{figa2} shows the results of the multi-BS joint environment sensing.
In order to make the total number of antennas deployed in the environment the same, we reduce the number of antennas on each BS.
The system parameters are set to the number of users $N_{\rm{u}}=20$, 5 BSs are deployed with $5 \times 1$ array antennas, and the SNR $E_{\rm{b}}/N_{\rm{ 0}}=20\rm{dB}$.
Compared with the single-BS scheme, as analyzed in (\ref{s6}) and (\ref{s7}), multi-BS joint environment sensing scheme effectively increases the environment sensing range. In Fig. \ref{figa2}, there is no scatterer that cannot be sensed even though the occlusion effect is exist.
It can be seen that the sensing results of baseline algorithms are more blurred than the results of the single-BS scheme.
This is because, as analyzed in Section \ref{gzzq}, when the total number of receiving antennas is the same, multi-BS joint environment sensing has a larger environment sensing range, which causes a small decrease in the accuracy of environment sensing within the sensing range.
  \begin{figure}[t]
    \centering
    \includegraphics[width=3.5in]{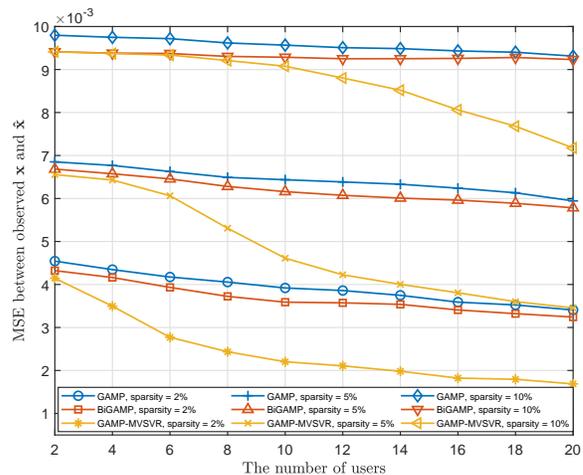}
    \caption{The relationship between the number of users and the MSE within the sensing range of the multi-BS scheme.}
    \label{figa4}
  \end{figure}
  
Fig. \ref{figa3} and Fig. \ref{figa4} show the relationship between the single-BS environment sensing performance, the multi-BS joint environment sensing performance and the number of users $N_{\rm{u}}$, respectively.
The MSE between the environmental information within the sensing range and the environmental information estimated by the algorithm reflects the accuracy of environment sensing within the sensing range.
The system parameters are set as follows: the single BS in the scenario in Fig. \ref{figa3} is deployed with a $5 \times 5$ array antenna, and 5 BSs in the scenario in Fig. \ref{figa4} are deployed with $5 \times 1$ array antennas, and the SNR $E_{\rm{b}}/N_{\rm{ 0}}=20\rm{dB}$.

\begin{remark}
 As analyzed in (\ref{s5}), as the number of users increases or the number of environmental scatterers decreases, the accuracy of environment sensing increases. Both the proposed algorithm and the baseline algorithm perform poorly when the number of users is small. But as the number of users increases, the performance of the proposed algorithm improves significantly.
\end{remark}

\begin{figure}[t]
  \centering
  \includegraphics[width=3.5in]{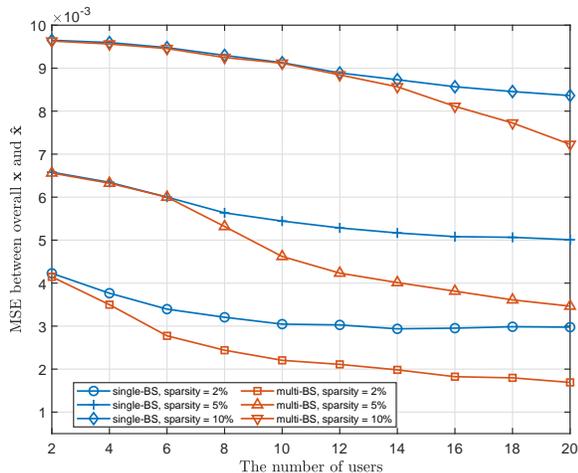}
  \caption{The relationship between the number of users and the MSE within the full range of two sensing schemes.}
  \label{figa5}
\end{figure}

  \begin{figure}[t]
    \centering
  \includegraphics[width=3.5in]{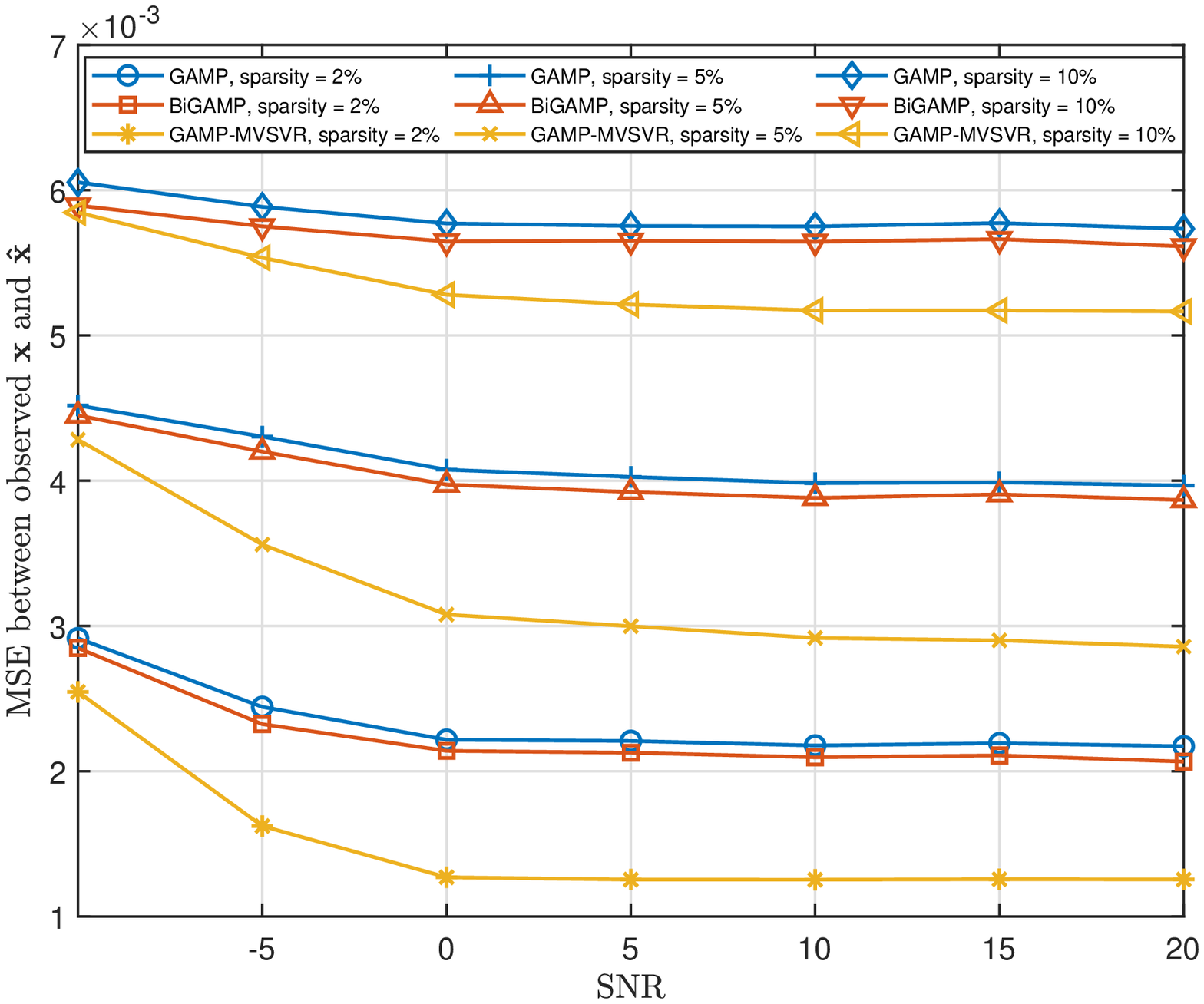}
  \caption{The relationship between the SNR and the MSE within the sensing range of the single-BS scheme.}
  \label{figa6}
  \end{figure}

  \begin{figure}[t]
    \centering
    \includegraphics[width=3.5in]{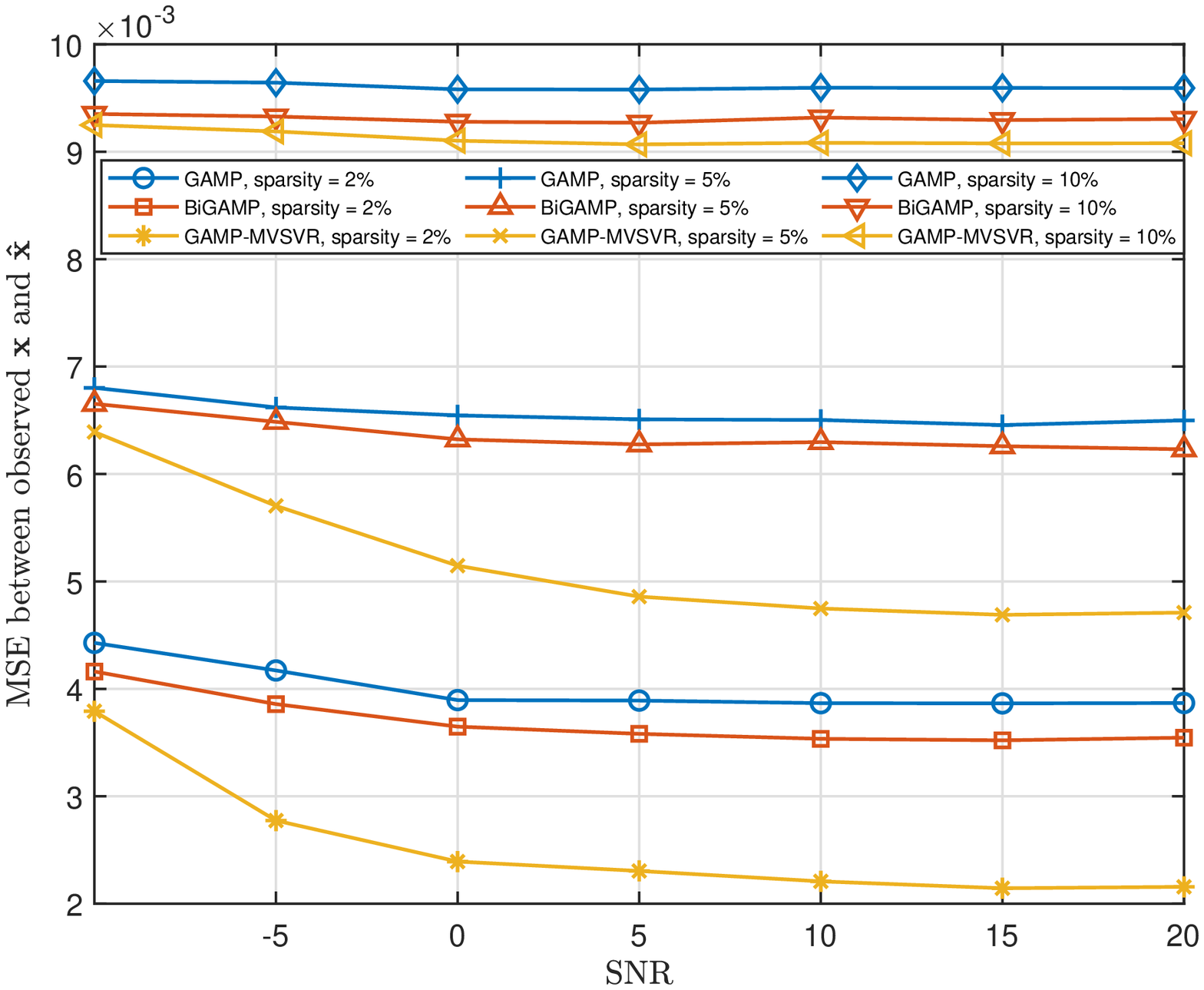}
    \caption{The relationship between the SNR and the MSE within the sensing range of the multi-BS scheme.}
    \label{figa7}
  \end{figure}

  \begin{figure}[t]
    \centering
  \includegraphics[width=3.5in]{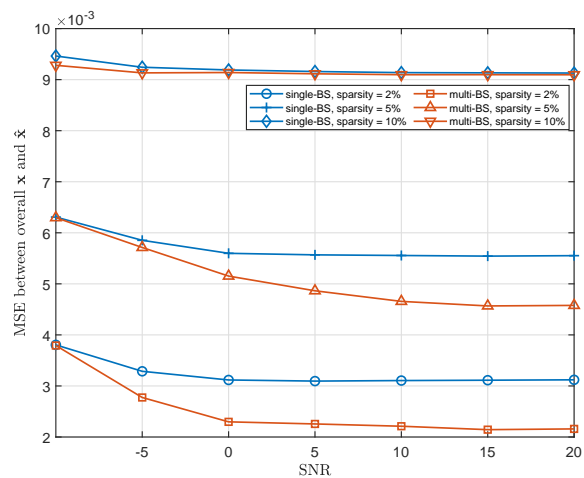}
  \caption{The relationship between the SNR and the MSE within the full range of two sensing schemes.}
  \label{figa8}
  \end{figure}
\begin{figure}[!t]
  \centering
  \includegraphics[width=3.5in]{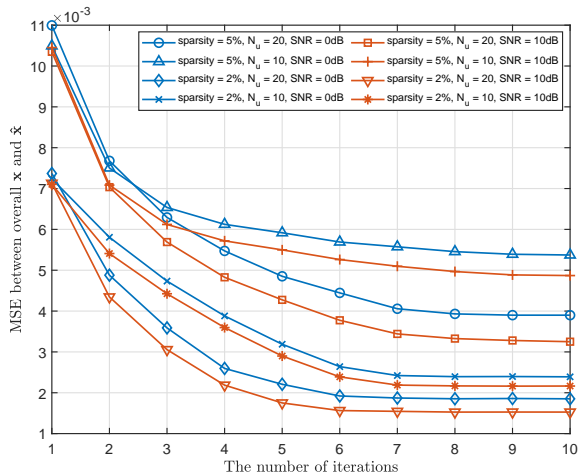}
  \caption{The convergence performance of proposed algorithm.}
  \label{figa9}
  \end{figure}

Fig. \ref{figa5} shows the relationship between the number of users and the MSE within the full range of two sensing schemes.
The MSE between the environmental information within the full range and the environmental information estimated by the algorithm reflects the size of the environment sensing range.
The system parameter settings of the two environment sensing schemes are the same as those in Fig. \ref{figa3} and Fig. \ref{figa4}, respectively.
\begin{remark}
 As analyzed in (\ref{s5}), as the number of users increases or the number of environmental scatterers decreases, the accuracy of environment sensing increases.
At the same time, as analyzed in (\ref{s6}) and (\ref{s7}), for the full range of the environment, the multi-BS scheme has a lower sensing error, so it effectively increases the environment sensing range.
\end{remark}

Fig. \ref{figa6} and Fig. \ref{figa7} show the relationship between the single-BS environment sensing performance, the multi-BS joint environment sensing performance and the SNR $E_{\rm{b}}/N_{\rm{0}}$, respectively.
The system parameters are set as follows: the single BS in the scenario in Fig. \ref{figa6} is deployed with a $5 \times 5$ array antenna, and 5 BSs in the scenario in Fig. \ref{figa7} are deployed with $5 \times 1$ array antennas, and the number of users $N_{\rm{u}}=10$.
\begin{remark}
 As the SNR increases, the estimated value ${{\bf{\hat H}}^{\rm{S}}}$ of the multipath channel is more accurate, and the accuracy of environment sensing increases. At the same time, the less sparse the environmental scatterer distribution, the closer the performance of the proposed algorithm and the baseline algorithm are. This is because the more environmental scatterers cause more occlusion, which reduces the performance of all algorithms at the same time.
\end{remark}

Fig. \ref{figa8} shows the relationship between the SNR and the MSE within the full range of two sensing schemes.
The system parameter settings of the two environment sensing schemes are the same as those in Fig. \ref{figa6} and Fig. \ref{figa7}, respectively.
\begin{remark}
 As the SNR increases, the accuracy of environment sensing increases. 
At the same time, as analyzed in (\ref{s6}) and (\ref{s7}), for the full range of the environment, the multi-BS scheme has a lower sensing error, so it effectively increases the environment sensing range.
\end{remark}

As shown in Fig. \ref{figa9}, we provide performance curves under various parameter settings to verify the convergence of the proposed algorithm. The antenna settings of the multi-base station scheme are the same as above.

\begin{remark}
 As the number of iterations increases, the sensing error MSE gradually decreases and converges. 
The sparser the environmental information, the faster the algorithm converges and the better the system performs.
Meanwhile, the increase in the number of users and the signal-to-noise ratio accelerates algorithm convergence and improves system performance.
\end{remark}

\section{Conclusion}
In this paper, we have designed a multi-view environment sensing scheme taking into full consideration the occlusion effect in an outdoor scenario.
Based on the occlusion effect, we have proposed a probabilistic reasoning model based on the factor graph. According to the AMP theory, the GAMP-MVSVR algorithm is proposed to achieve multi-view environment sensing. 
The proposed GAMP-MVSVR algorithm achieves environment sensing by iteratively estimating the scattering coefficients of the cloud points and their occlusion relationship.
In each iteration, the occlusion relationship between the cloud points of the sparse environment is recalculated according to the proposed occlusion detection rule, and in turn, used to estimate the scattering coefficients of the cloud points.
Our proposed algorithm achieves improved sensing performance with multi-BS collaboration in addition to the multi-views from the users, which will provide ideas for future ISAC system design.

\begin{appendices}
  \section{}\label{fld}
According to the factor graph shown in Fig. \ref{fig5}, the logarithmic form of the passed message is calculated as follows. In the $t$-th iteration, the message ${\mu _{\hat h_{{n_{\rm{H}}}}^{\rm{S}} \to {x_{{n_{\rm{s}}}}}}}\left( {t,{x_{{n_{\rm{s}}}}}} \right)$ passed from the factor node $p_{{{\bf{\hat H}}^{\rm{S}}}|{{\bf{ H}}^{\rm{S}}}}( {{{\bf{\hat H}}^{\rm{S}}}|{{\bf{H}}^{\rm{S}}}} )$ to the variable node $\bm{x}$ is 
{\small \begin{align}
  &{\mu _{\hat h_{{n_{\rm{H}}}}^{\rm{S}} \to {x_{{n_{\rm{s}}}}}}}\left( {t,{x_{{n_{\rm{s}}}}}} \right)\nonumber\\ &= \log \int_{\left\{ {{{\tilde h}_{{n_{\rm{H}}},k}}} \right\}_{k = 1}^{{N_{\rm{s}}}},{{\left\{ {{x_r}} \right\}}_{r \ne {n_{\rm{s}}}}}} {{p_{\hat h_{{n_{\rm{H}}}}^{\rm{S}}|h_{{n_{\rm{H}}}}^{\rm{S}}}}\left( {\hat h_{{n_{\rm{H}}}}^{\rm{S}}\left| {\sum\limits_{k = 1}^{N_{\rm{s}}} {{{\tilde h}_{{n_{\rm{H}}},k}}{x_k}} } \right.} \right)} \nonumber\\ & \times \prod\limits_{r \ne {n_{\rm{s}}}} {\exp \left( {{\mu _{{x_r} \to \hat h_{{n_{\rm{H}}}}^{\rm{S}}}}\left( {t,{x_r}} \right)} \right)}\nonumber\\ & \times \prod\limits_{k = 1}^{{N_{\rm{s}}}} {\exp \left( {{\mu _{{{\tilde h}_{{n_{\rm{H}}},k}} \to \hat h_{{n_{\rm{H}}}}^{\rm{S}}}}\left( {t,{{\tilde h}_{{n_{\rm{H}}},k}}} \right)} \right)}  + c, \label{a6} 
  \end{align}}
where the passed message $\mu$ is expressed in logarithmic form. $c$ is a normalization constant, so that the integral of the exponential passed message $\exp ({\mu _{\hat h_{{n_{\rm{H}}}}^{\rm{S}} \to {x_{{n_{\rm{s}}}}}}}) $ is 1. In the subsequent passed message and marginal distribution in this section, the constant $c$ has the same meaning.
In the $(t+1)$-th iteration, the message ${{\mu _{{x_r} \to \hat h_{{n_{\rm{H}}}}^{\rm{S}}}}\left( {t+1,{x_r}} \right)}$ passed from the variable node $\bm{x}$ to the factor node $p_{{{\bf{\hat H}}^{\rm{S}}}|{{\bf{ H}}^{\rm{S}}}}( {{{\bf{\hat H}}^{\rm{S}}}|{{\bf{H}}^{\rm{S}}}} )$ is 
\begin{equation}\small
  \begin{split}
  {\mu _{{x_{{n_{\rm{s}}}}} \to \hat h_{{n_{\rm{H}}}}^{\rm{S}}}} \left( {t + 1,{x_{{n_{\rm{s}}}}}} \right)&\\ = \log {p_x}\left( {{x_{{n_{\rm{s}}}}}} \right) &+ \sum\limits_{k \ne {n_{\rm{H}}}} {{\mu _{\hat h_k^{\rm{S}} \to {x_{{n_{\rm{s}}}}}}}\left( {t,{x_{{n_{\rm{s}}}}}} \right)}  + c. \label{a7} 
  \end{split}
\end{equation}

Let the mean and variance of the exponential form of the message passed in (\ref{a7}) be ${\bar x_{{n_{\rm{s}}},{n_{\rm{H}}}}}\left( {t + 1} \right)$ and $\sigma _{{n_{\rm{s}}},{n_{\rm{H}}}}^{\rm{x}}\left( {t + 1} \right)$, respectively.
In the $t$-th iteration, the message ${\mu _{\hat h_{{n_{\rm{H}}}}^{\rm{S}} \to {{\tilde h}_{{n_{\rm{H}}},{n_{\rm{s}}}}}}}( {t,{{\tilde h}_{{n_{\rm{H}}},{n_{\rm{s}}}}}} )$ passed from the factor node $p_{{{\bf{\hat H}}^{\rm{S}}}|{{\bf{ H}}^{\rm{S}}}}( {{{\bf{\hat H}}^{\rm{S}}}|{{\bf{H}}^{\rm{S}}}} )$ to the variable node $\bf{\tilde H}$ is expressed as 
  {\small \begin{align}
      &{\mu _{\hat h_{{n_{\rm{H}}}}^{\rm{S}} \to {{\tilde h}_{{n_{\rm{H}}},{n_{\rm{s}}}}}}}\left( {t,{{\tilde h}_{{n_{\rm{H}}},{n_{\rm{s}}}}}} \right)\nonumber\\ & = \log \int_{{{\left\{ {{{\tilde h}_{{n_{\rm{H}}},r}}} \right\}}_{r \ne {n_{\rm{s}}}}},\left\{ {{x_k}} \right\}_{k = 1}^{{N_{\rm{s}}}}} {{p_{\hat h_{{n_{\rm{H}}}}^{\rm{S}}|h_{{n_{\rm{H}}}}^{\rm{S}}}}\left( {\hat h_{{n_{\rm{H}}}}^{\rm{S}}\left| {\sum\limits_{k = 1}^{{N_{\rm{s}}}} {{{\tilde h}_{{n_{\rm{H}}},k}}{x_k}} } \right.} \right)} \nonumber\\
      &\times \prod\limits_{k = 1}^{{N_{\rm{s}}}} {\exp \left( {{\mu _{{x_k} \to \hat h_{{n_{\rm{H}}}}^{\rm{S}}}}\left( {t,{x_k}} \right)} \right)} \nonumber\\ &\times \prod\limits_{r \ne {n_{\rm{s}}}} {\exp \left( {{\mu _{{{\tilde h}_{{n_{\rm{H}}},r}} \to \hat h_{{n_{\rm{H}}}}^{\rm{S}}}}\left( {t,{{\tilde h}_{{n_{\rm{H}}},r}}} \right)} \right)}  + c. \label{a8} 
      \end{align}}

In the $(t+1)$-th iteration, similar to \eqref{a7}, the message ${\mu _{{{\tilde h}_{{n_{\rm{H}}},{n_{\rm{s}}}}} \to \hat h_{{n_{\rm{H}}}}^{\rm{S}}}}( {t + 1,{{\tilde h}_{{n_{\rm{H}}},{n_{\rm{s}}}}}} )$ passed from the variable node $\bf{\tilde H}$ to the factor node $p_{{{\bf{\hat H}}^{\rm{S}}}|{{\bf{ H}}^{\rm{S}}}}( {{{\bf{\hat H}}^{\rm{S}}}|{{\bf{H}}^{\rm{S}}}} )$ is expressed as 
  {\small \begin{align}
  &{\mu _{{{\tilde h}_{{n_{\rm{H}}},{n_{\rm{s}}}}} \to \hat h_{{n_{\rm{H}}}}^{\rm{S}}}}\left( {t + 1,{{\tilde h}_{{n_{\rm{H}}},{n_{\rm{s}}}}}} \right) = \log {p_{{{\tilde h}_{{n_{\rm{H}}},{n_{\rm{s}}}}}|{\bm{x}}}}\left( {{{\tilde h}_{{n_{\rm{H}}},{n_{\rm{s}}}}}|{\bm{x}}} \right) \nonumber\\ &+ \sum\limits_{k \ne {n_{\rm{H}}}} {{\mu _{\hat h_k^{\rm{S}} \to {{\tilde h}_{{n_{\rm{H}}},{n_{\rm{s}}}}}}}\left( {t,{{\tilde h}_{{n_{\rm{H}}},{n_{\rm{s}}}}}} \right)}  + c. \label{a9} 
  \end{align}}

Let the mean and variance of the exponential form of the message passed in (\ref{a9}) be ${\bar h_{{n_{\rm{H}}}{n_{\rm{s}}},{n_{\rm{H}}}}}\left( {t + 1} \right)$ and $\sigma _{{n_{\rm{H}}}{n_{\rm{s}}},{n_{\rm{H}}}}^{\rm{h}}\left( {t + 1} \right)$, respectively.

According to the factor graph in Fig. \ref{fig5}, the logarithmic marginal distribution of the variable $\bf{\tilde H}$ and $\bm{x}$ are calculated as follows. With the iteration of the algorithm, the marginal distributions of variables gradually converge to their MMSE estimates. The marginal distribution of the variable $\bm{x}$ is  
\begin{equation}\small
  \begin{split}
  {\mu _{{x_{{n_{\rm{s}}}}}}}\left( {t + 1,{x_{{n_{\rm{s}}}}}} \right)& \\= \log {p_x}&\left( {{x_{{n_{\rm{s}}}}}} \right) + \sum\limits_{k = 1}^{{N_{\rm{H}}}} {{\mu _{\hat h_k^{\rm{S}} \to {x_{{n_{\rm{s}}}}}}}\left( {t,{x_{{n_{\rm{s}}}}}} \right)}  + c. \label{a10} 
  \end{split}
  \end{equation}
  
Let the mean and variance of the exponential form of the marginal distribution in (\ref{a10}) be ${\hat x_{{n_{\rm{s}}}}}\left( {t + 1} \right)$ and $\sigma _{{n_{\rm{s}}}}^{\rm{x}}\left( {t + 1} \right)$, respectively. The marginal distribution of the variable $\bf{\tilde H}$ is  
\begin{equation}\small
  \begin{split}
  &{\mu _{{{\tilde h}_{{n_{\rm{H}}},{n_{\rm{s}}}}}}}\left( {t + 1,{{\tilde h}_{{n_{\rm{H}}},{n_{\rm{s}}}}}} \right) = \log {p_{{{\tilde h}_{{n_{\rm{H}}},{n_{\rm{s}}}}}|{\bm{x}}}}\left( {{{\tilde h}_{{n_{\rm{H}}},{n_{\rm{s}}}}}|{\bm{x}}} \right) \\ &+ \sum\limits_{k = 1}^{{N_{\rm{H}}}} {{\mu _{\hat h_k^{\rm{S}} \to {{\tilde h}_{{n_{\rm{H}}},{n_{\rm{s}}}}}}}\left( {t,{{\tilde h}_{{n_{\rm{H}}},{n_{\rm{s}}}}}} \right)}  + c. \label{a11} 
  \end{split}
  \end{equation}
  
Let the mean and variance of the exponential form of the marginal distribution in (\ref{a11}) be ${\hat h_{{n_{\rm{H}}},{n_{\rm{s}}}}}\left( {t + 1} \right)$ and $\sigma _{{n_{\rm{H}}},{n_{\rm{s}}}}^{\rm{h}}\left( {t + 1} \right)$, respectively.

  \section{}\label{flb}
Here we derive the approximate result (\ref{a17}).
The Roman font is used to represent the elements ${{\rm{h}}_{n_{\rm{H}}}^{\rm{S}}}$, $\rm{\tilde h}_{{n_{\rm{H}}},{n_{\rm{s}}}}$, and $\rm{x}_{n_{\rm{s}}}$ in the matrix ${{\bf{H}}^{\rm{S}}}$, $\bf{\tilde H}$, and $\bm{x}$ in the case of CLT.
Let $\Delta {h_{{n_{\rm{H}}}{n_{\rm{s}}},{n_{\rm{H}}}}} = { \rm{\tilde h}_{{n_{\rm{H}}},{n_{\rm{s}}}}} - {\bar h_{{n_{\rm{H}}}{n_{\rm{s}}},{n_{\rm{H}}}}}\left( t \right)$ and $\Delta {x_{{n_{\rm{s}}},{n_{\rm{H}}}}} = {{\rm{x}}_{{n_{\rm{s}}}}} - {\bar x_{{n_{\rm{s}}},{n_{\rm{H}}}}}\left( t \right)$, then the mean is $0$, and the variance is $\sigma _{{n_{\rm{H}}}{n_{\rm{s}}},{n_{\rm{H}}}}^{\rm{h}}\left( {t} \right)$ and $\sigma _{{n_{\rm{s}}}}^{\rm{x}}\left( {t} \right)$.

In the $t$ iteration, the sparsity $\lambda$ in the prior distribution of environmental information ${p_{\bm{x}}}\left( {\bm{x}} \right)$ remains constant. As mentioned in Section \ref{yz2bl}, we consider that variable $\bm{x}$ is independent of variable $\bf{\tilde H}$ in the $t$-th iteration.
Therefore, the mean ${\hat h_{{n_{\rm{H}}},{n_{\rm{s}}}}}$ and the variance $\sigma _{{n_{\rm{H}}},{n_{\rm{s}}}}^{\rm{h}}$ of variable $\bf{\tilde H}$ are independent of variable $\bm{x}$.
${{\bf{H}}^{\rm{S}}}$ is expressed as
  {\small \begin{align}
      \rm{h}_{{n_{\rm{H}}}}^{\rm{S}}& = \sum\limits_{k = 1}^N {{{{\rm{\tilde h}}}_{{n_{\rm{H}}},{n_{\rm{s}}}}}{{\rm{x}}_{{n_{\rm{s}}}}}} \nonumber \\ &= \left( {\Delta {{\hat h}_{{n_{\rm{H}}}{n_{\rm{s}}},{n_{\rm{H}}}}} + {{\bar h}_{{n_{\rm{H}}}{n_{\rm{s}}},{n_{\rm{H}}}}}\left( t \right)} \right){{\rm{x}}_{{n_{\rm{s}}}}} \nonumber\\
      &  + \sum\limits_{k \ne {n_{\rm{s}}}} \left( {{{\bar h}_{{n_{\rm{H}}}k,{n_{\rm{H}}}}}\left( t \right){{\bar x}_{k,{n_{\rm{H}}}}}\left( t \right) + {{\bar h}_{{n_{\rm{H}}}k,{n_{\rm{H}}}}}\left( t \right)\Delta {x_{k,{n_{\rm{H}}}}}}\right. \nonumber\\ &+ \left.{\Delta {{ h}_{{n_{\rm{H}}}k,{n_{\rm{H}}}}}{{\bar x}_{k,{n_{\rm{H}}}}}\left( t \right) + \Delta {{ h}_{{n_{\rm{H}}}k,{n_{\rm{H}}}}}\Delta {x_{k,{n_{\rm{H}}}}}} \right).\label{a12}
      \end{align}}
      
Based on the CLT condition, ${{\rm{x}}_{{n_{\rm{s}}}}} = {x_{{n_{\rm{s}}}}}$, the mean and variance of $\rm{h}_{{n_{\rm{H}}}}^{\rm{S}}$ in (\ref{a12}) is written as
\begin{equation}\small
  {\mathbb{E}}\left\{ {{\rm{h}}_{{n_{\rm{H}}}}^{\rm{S}}} \right\} = {\bar h_{{n_{\rm{H}}}{n_{\rm{s}}},{n_{\rm{H}}}}}\left( t \right){x_{{n_{\rm{s}}}}} + {\bar p_{{n_{\rm{s}}},{n_{\rm{H}}}}}\left( t \right),\label{a13}
  \end{equation}
  \begin{equation}\small
  {\mathbb{D}}\left\{ {{\rm{h}}_{{n_{\rm{H}}}}^{\rm{S}}} \right\} = \sigma _{{n_{\rm{H}}}{n_{\rm{s}}},{n_{\rm{H}}}}^{\rm{h}}\left( t \right)x_{{n_{\rm{s}}}}^2 + \sigma _{{n_{\rm{s}}},{n_{\rm{H}}}}^{\rm{p}}\left( t \right),\label{a14}
  \end{equation}
where
\begin{equation}\small
  {\bar p_{{n_{\rm{s}}},{n_{\rm{H}}}}}\left( t \right) = \sum\limits_{k \ne {n_{\rm{s}}}} {{{\bar h}_{{n_{\rm{H}}}k,{n_{\rm{H}}}}}\left( t \right){{\bar x}_{k,{n_{\rm{H}}}}}\left( t \right)},\label{a15}
  \end{equation}
 \begin{equation}\small
  \begin{split}
  &\sigma _{{n_{\rm{s}}},{n_{\rm{H}}}}^{\rm{p}}\left( t \right)  = \sum\limits_{k \ne {n_{\rm{s}}}} {\left( {\bar h_{{n_{\rm{H}}}k,{n_{\rm{H}}}}^2\left( t \right)\sigma _{k,{n_{\rm{H}}}}^{\rm{x}}\left( t \right) }\right.}\\ & {\left.{+ \sigma _{{n_{\rm{H}}}k,{n_{\rm{H}}}}^{\rm{h}}\left( t \right)\bar x_{k,{n_{\rm{H}}}}^2\left( t \right) + \sigma _{{n_{\rm{H}}}k,{n_{\rm{H}}}}^{\rm{h}}\left( t \right)\sigma _{k,{n_{\rm{H}}}}^{\rm{x}}\left( t \right)} \right)}.\label{a16}
  \end{split}
  \end{equation}

Under CLT conditions, we approximate the last two product terms of message ${\mu _{\hat h_{{n_{\rm{H}}}}^{\rm{S}} \to {x_{{n_{\rm{s}}}}}}}\left( {t,{x_{{n_{\rm{s}}}}}} \right)$ in (\ref{a6}) to the marginal distribution of variable $\rm{h}_{{n_{\rm{H}}}}^{\rm{S}}$ in (\ref{a12}).
Therefore, according to the conditional Gaussian approximation, we obtain (\ref{a17}).

\section{}\label{flc}
Here we derive the approximate result (\ref{a26}).
Pluging (\ref{a19}) and (\ref{a20}) into (\ref{a17}) we obtain
  {\small \begin{align}
      & {F_{{n_{\rm{H}}}}}\left( {{{\bar h}_{{n_{\rm{H}}}{n_{\rm{s}}},{n_{\rm{H}}}}}\left( t \right){x_{{n_{\rm{s}}}}} + {{\bar p}_{{n_{\rm{s}}},{n_{\rm{H}}}}}\left( t \right),}\right. \nonumber\\ &\quad \left.{ \sigma _{{n_{\rm{H}}}{n_{\rm{s}}},{n_{\rm{H}}}}^{\rm{h}}\left( t \right)x_{{n_{\rm{s}}}}^2 + \sigma _{{n_{\rm{s}}},{n_{\rm{H}}}}^{\rm{p}}\left( t \right);\hat h_{{n_{\rm{H}}}}^{\rm{S}}} \right) + c\nonumber\\ &
       \approx {F_{{n_{\rm{H}}}}}\left( {{{\bar h}_{{n_{\rm{H}}}{n_{\rm{s}}},{n_{\rm{H}}}}}\left( t \right)\left( {{x_{{n_{\rm{s}}}}} - {{\hat x}_{{n_{\rm{s}}}}}\left( t \right)} \right) + {{\bar p}_{{n_{\rm{H}}}}}\left( t \right),} \right.\nonumber\\ & \quad
      \left. {\sigma _{{n_{\rm{H}}}{n_{\rm{s}}},{n_{\rm{H}}}}^{\rm{h}}\left( t \right)\left( {x_{{n_{\rm{s}}}}^2 - \hat x_{{n_{\rm{s}}}}^2\left( t \right)} \right) + \sigma _{{n_{\rm{H}}}}^{\rm{p}}\left( t \right);\hat h_{{n_{\rm{H}}}}^{\rm{S}}} \right) + c,\label{a21}
      \end{align}}
where, since $N_{\rm{s}}\rightarrow \infty $, we drop infinitesimal terms such as ${{\bar h}_{{n_{\rm{H}}}{n_{\rm{s}}},{n_{\rm{H}}}}}\left( t \right)\left( {{{\hat x}_{{n_{\rm{s}}}}}\left( t \right) - {{\bar x}_{{n_{\rm{s}}},{n_{\rm{H}}}}}\left( t \right)} \right)$ and so on.

As with most AMP algorithms \cite{Donoho2,Rangan,Jason}, we make the same derivation. We perform Taylor series expansion of variable ${x_{{n_{\rm{s}}}}}$ in (\ref{a21}) at point ${{{\hat x}_{{n_{\rm{s}}}}}\left( t \right)}$, the passed message in (\ref{a17}) is approximated as
  {\small \begin{align}
      &{\mu _{\hat h_{{n_{\rm{H}}}}^{\rm{S}} \to {x_{{n_{\rm{s}}}}}}}\left( {t,{x_{{n_{\rm{s}}}}}} \right) \approx c + {F_{{n_{\rm{H}}}}}\left( {{{\bar p}_{{n_{\rm{H}}}}}\left( t \right),\sigma _{{n_{\rm{H}}}}^{\rm{p}}\left( t \right);\hat h_{{n_{\rm{H}}}}^{\rm{S}}} \right) \nonumber\\
      &+ {{\bar h}_{{n_{\rm{H}}}{n_{\rm{s}}},{n_{\rm{H}}}}}\left( t \right)\left( {{x_{{n_{\rm{s}}}}} - {{\hat x}_{{n_{\rm{s}}}}}\left( t \right)} \right) \nonumber\\
      &\quad \quad \quad \quad \times{{F'}_{{n_{\rm{H}}},1}}\left( {{{\bar p}_{{n_{\rm{H}}}}}\left( t \right),\sigma _{{n_{\rm{H}}}}^{\rm{p}}\left( t \right);\hat h_{{n_{\rm{H}}}}^{\rm{S}}} \right)\nonumber\\
      &+ 2\sigma _{{n_{\rm{H}}}{n_{\rm{s}}},{n_{\rm{H}}}}^{\rm{h}}\left( t \right){{\hat x}_{{n_{\rm{s}}}}}\left( t \right)\left( {{x_{{n_{\rm{s}}}}} - {{\hat x}_{{n_{\rm{s}}}}}\left( t \right)} \right)\nonumber\\
      &\quad \quad \quad \quad \times{{F'}_{{n_{\rm{H}}},2}}\left( {{{\bar p}_{{n_{\rm{H}}}}}\left( t \right),\sigma _{{n_{\rm{H}}}}^{\rm{p}}\left( t \right);\hat h_{{n_{\rm{H}}}}^{\rm{S}}} \right)\nonumber\\
      & + \frac{1}{2}\bar h_{{n_{\rm{H}}}{n_{\rm{s}}},{n_{\rm{H}}}}^2\left( t \right){\left( {{x_{{n_{\rm{s}}}}} - {{\hat x}_{{n_{\rm{s}}}}}\left( t \right)} \right)^2}\nonumber\\
      &\quad \quad \quad \quad  \times{{F''}_{{n_{\rm{H}}},1}}\left( {{{\bar p}_{{n_{\rm{H}}}}}\left( t \right),\sigma _{{n_{\rm{H}}}}^{\rm{p}}\left( t \right);\hat h_{{n_{\rm{H}}}}^{\rm{S}}} \right)\nonumber\\
      &  + \sigma _{{n_{\rm{H}}}{n_{\rm{s}}},{n_{\rm{H}}}}^{\rm{h}}\left( t \right){\left( {{x_{{n_{\rm{s}}}}} - {{\hat x}_{{n_{\rm{s}}}}}\left( t \right)} \right)^2}\nonumber\\
      &\quad \quad \quad \quad \times{{F'}_{{n_{\rm{H}}},2}}\left( {{{\bar p}_{{n_{\rm{H}}}}}\left( t \right),\sigma _{{n_{\rm{H}}}}^{\rm{p}}\left( t \right);\hat h_{{n_{\rm{H}}}}^{\rm{S}}} \right).\label{a22}
      \end{align}}
where ${F'_{{n_{\rm{H}}},1}}$ and ${F''_{{n_{\rm{H}}},1}}$ represent the first and second derivatives of the first argument of ${F_{{n_{\rm{H}}}}}$, and ${F'_{{n_{\rm{H}}},2}}$ represents the first derivative of the second argument of ${F_{{n_{\rm{H}}}}}$.
We drop infinitesimal terms such as the second derivative term of the second argument of ${F_{{n_{\rm{H}}}}}$.

According to the derivation in Appendix \ref{fla} that
\begin{equation}\small
  {F'_{n_{\rm{H}},2}} = \frac{1}{2}\left[ {{{\left( {{F'_{n_{\rm{H}},1}}} \right)}^2} + {F''_{n_{\rm{H}},1}}} \right].\label{a25}
  \end{equation}

We plug (\ref{a23}), (\ref{a24}) and (\ref{a25}) into (\ref{a22}), and absorb the constant term into the constant $c$, we obtain (\ref{a26}), 
where under the CLT condition, we use ${\hat h_{{n_{\rm{H}}},{n_{\rm{s}}}}^2\left( t \right)}$ instead of $\bar h_{{n_{\rm{H}}}{n_{\rm{s}}},{n_{\rm{H}}}}^2\left( t \right)$, and ${\sigma _{{n_{\rm{s}}},{n_{\rm{H}}}}^{\rm{h}}\left( t \right)}$ instead of $\sigma _{{n_{\rm{H}}}{n_{\rm{s}}},{n_{\rm{H}}}}^{\rm{h}}\left( t \right)$.

\section{}\label{fla}
Here we derive the equation (\ref{a25}).
According to (\ref{a18}), we know the following
\begin{equation}\small
  \frac{{\partial \exp \left( {{F_{{n_{\rm{H}}}}}} \right)}}{{\partial {\sigma ^p}}} = {F'_{{n_{\rm{H}}},2}}\exp \left( {{F_{{n_{\rm{H}}}}}} \right),\label{fl3}
  \end{equation}
\begin{equation}\small
  \frac{{{\partial ^2}\exp \left( {{F_{{n_{\rm{H}}}}}} \right)}}{{\partial {p^2}}} = \left[ {{{\left( {{F'_{{n_{\rm{H}}},1}}} \right)}^2} + {F''_{{n_{\rm{H}}},1}}} \right]\exp \left( {{F_{{n_{\rm{H}}}}}} \right).\label{fl4}
  \end{equation}

According to (\ref{a18}), we obtain
{\small \begin{align}
    & \frac{{\partial \exp \left( {{F_{{n_{\rm{H}}}}}} \right)}}{{\partial {\sigma ^{\rm{p}}}}} \nonumber\\
      & = \int_{{h^{\rm{S}}}} {{p_{{{\hat h}^{\rm{S}}}|{h^{\rm{S}}}}}\left( {{{\hat h}^{\rm{S}}}|{h^{\rm{S}}}} \right)\frac{1}{{\sqrt {2\pi } }}\left[ { - \frac{1}{2}{{\left( {{\sigma ^{\rm{p}}}} \right)}^{{\rm{ - 3/2}}}} }\right.}\nonumber\\ & {\left.{- {{\left( {{\sigma ^{\rm{p}}}} \right)}^{{\rm{ - 5/2}}}}\frac{{{{\left( {{h^{\rm{S}}} - p} \right)}^2}}}{2}} \right]\exp \left( {\frac{{{{\left( {{h^{\rm{S}}} - p} \right)}^2}}}{{2{\sigma ^{\rm{p}}}}}} \right)} \nonumber\\
      & = \frac{1}{2}\frac{{{\partial ^2}\exp \left( {{F_{{n_{\rm{H}}}}}} \right)}}{{\partial {p^2}}}
    .\label{fl2}
  \end{align}}

Plugging (\ref{fl3}) and (\ref{fl4}) into (\ref{fl2}), we obtain (\ref{a25}).

\end{appendices}

\ifCLASSOPTIONcaptionsoff
  \newpage
\fi

\bibliographystyle{IEEEbib}
\bibliography{ref_short}

\end{document}